\documentclass[12pt,a4paper]{amsart}
\pdfoutput=1
\usepackage{graphicx}
\usepackage{stmaryrd}
\usepackage{amssymb}
\usepackage{amsmath}
\usepackage{amsthm}
\usepackage{dsfont}
\usepackage{hyperref}
\usepackage{amsfonts}
\usepackage{amsrefs}
\usepackage{epsfig}
\usepackage{latexsym}
\usepackage{cite}

\usepackage{amsthm}
\usepackage{enumerate}
\usepackage{centernot}
\usepackage{fullpage}
\usepackage{mathrsfs}

\usepackage{color}

\newtheorem{theorem}{Theorem}[section]
\newtheorem{corollary}[theorem]{Corollary}

\newtheorem{lemma}[theorem]{Lemma}

\theoremstyle{definition}

\newtheorem{remark}[theorem]{Remark}

\numberwithin{equation}{section}

\newcommand{\diam}{\text{\textnormal{diam}}}

\newcommand{\supp}{\text{\textnormal{supp}}}

\begin{document}
\title{Conformal Covariance of Connection Probabilities \\ and Fields in 2D Critical Percolation}

\author{Federico Camia}

\address{Division of Science, NYU Abu Dhabi, Saadiyat Island, Abu Dhabi, UAE \& Courant Institute of Mathematical Sciences, New York University,
	251 Mercer st, New York, NY 10012, USA.}
\email{federico.camia@nyu.edu}


\begin{abstract}
Fitting percolation into the conformal field theory framework requires showing that connection probabilities have a conformally invariant scaling limit.
For critical site percolation on the triangular lattice, we prove that the probability that $n$ vertices belong to the same open cluster has a well-defined scaling limit for every $n \geq 2$. Moreover, the limiting functions $P_n(x_1,\ldots,x_n)$ transform covariantly under M\"obius transformations of the plane as well as under local conformal maps, i.e., they behave like correlation functions of primary operators in conformal field theory. In particular, they are invariant under translations, rotations and inversions, and $P_n(sx_1,\ldots,sx_n)=s^{-5n/48}P_n(x_1,\ldots,x_n)$ for any $s>0$.
This implies that 
$P_{2}(x_1,x_2)=C_2 \Vert x_1-x_2 \Vert^{-5/24}$ and $P_3(x_1,x_2,x_3) = C_3 \Vert x_1-x_2 \Vert^{-5/48} \Vert x_1-x_3 \Vert^{-5/48} \Vert x_2-x_3 \Vert^{-5/48}$, for some constants $C_2$ and $C_3$. 

We also define a site-diluted spin model whose $n$-point correlation functions $\mathrm{C}_{n}$ can be expressed in terms of percolation connection probabilities and, as a consequence, have a well-defined scaling limit with the same properties as the functions $P_n$.
In particular, $\mathrm{C}_{2}(x_1,x_2)=P_{2}(x_1,x_2)$. We prove that the magnetization field associated with this spin model has a well-defined scaling limit in an appropriate space of distributions. The limiting field transforms covariantly under M\"obius transformations with exponent (scaling dimension) $5/48$. A heuristic analysis of the four-point function of the magnetization field suggests the presence of an additional conformal field of scaling dimension $5/4$, which counts the number of percolation four-arm events and can be identified with the so-called ``four-leg operator'' of conformal field theory.
%
\end{abstract}


\maketitle

\section{Introduction}

\subsection{Background and motivation}
Percolation was introduced by Broadbent and Hammersley to model the spread of a gas or a fluid through a porous medium \cite{bh}. The model consists essentially of a graph or a lattice (e.g., the square, triangular or hexagonal lattice in two dimensions) in which edges or vertices are declared open (occupied/present) or closed (vacant/absent) at random, which generates a random version of the original graph or lattice. Percolation theory consists in the study of the connectivity properties of this random graph or lattice.

Percolation has been extensively studied by both physicists and mathematicians and has a large number of applications (see, e.g., \cite{br-book,grimmett-book,kesten-book,sa-book}). In dimensions higher than one, the model undergoes a geometric phase transition, resulting in an abrupt change of its connectivity properties at a critical density.  The two-dimensional version of the model is particularly well understood (see \cite{br-book,grimmett-book,kesten-book,sa-book}), including at the critical density, where the large scale properties are believed to be described by a conformal field theory (see, e.g., \cite{DFMS,Henkel}).

The question of conformal invariance in critical percolation has played a crucial role in the rapid development of the mathematical theory of scaling limits in two dimensions which has taken place in the last twenty-five years. Indeed, percolation, together with models like the uniform spanning tree and the Ising model, has been a laboratory where crucial ideas and tools have been developed \cite{Sch00,LSW02,LSW04,Smi01}.

The hypothesis of conformal invariance in critical systems goes back to the work of Polyakov and collaborators \cite{Polyakov70,BPZ84a,BPZ84b} and is usually expressed in terms of correlation functions of some observable ``field'' of the system. In the case of percolation, the formulation of this hypothesis was not as straightforward as for other models of statistical mechanics due to the purely geometric nature of percolation and the lack of a natural field analogous to the magnetization field of the Ising and Potts models.

The hypothesis that crossing probabilities should have a conformally invariant scaling limit, as the lattice spacing is sent to zero, is attributed to Michael Aizenman in \cite{LPS-A94}, an influential article, published in 1994, which brought the problem of conformal invariance in percolation to the attention of the mathematics community.
A couple of years earlier, applying non-rigorous ideas from conformal field theory, Cardy \cite{Cardy92} had obtained a conformally-invariant formula for the scaling limit of crossing probabilities. According to Langland, Pouliot and Saint-Aubin \cite{LPS-A94}, Cardy was motivated by Aizenman's hypothesis.

The first proof of conformal invariance was obtained for site percolation on the triangular lattice by Smirnov \cite{Smi01}, who showed that crossing probabilities have a conformally invariant scaling limit given by Cardy's formula \cite{Cardy92}.

Building on Smirnov's result and on the introduction of the Schramm-Loewner Evolution (SLE) by Schramm \cite{Sch00}, Newman and the present author \cite{CN04,CN06} showed that the collection of critical percolation interfaces converges in the scaling limit to a collection of non-simple, non-crossing loops whose distribution is invariant under conformal transformations. In doing so, they provided the first construction of a non-simple, nested conformal loop ensemble (CLE), and the first proof that such a CLE can be obtained from the scaling limit of a critical model of statistical mechanics.  The concept of conformal loop ensemble was later formulated in full generality by Sheffield \cite{She09}, and has been extensively studied, largely because conformal loop ensembles are conjectured to describe the scaling limit of critical interfaces in various two-dimensional models of statistical mechanics. (There are too many articles on CLE to list them all---see \cite{SW12} as an important example.) The scaling limit of percolation corresponds to CLE$_6$ \cite{CN08}, an ensemble of loops locally distributed like SLE$_6$, the Schramm-Loewner Evolution with parameter $\kappa=6$ \cite{CN07}.

In the physics literature, percolation is usually investigated as a (non-unitary) conformal field theory (CFT) with central charge $c=0$. 
It is also a prototypical example of a logarithmic field theory \cite{VJS12,CR13}, so that its study belongs to an area of research very active both in physics and mathematics (see, e.g., \cite{CR13} and \cite{LPW21}). As such, percolation is often studied by analyzing connection probabilities (sometimes called connectivity functions). In lattice site percolation, these are defined as the probabilities $P^a_n(x^a_1,\ldots,x^a_n)$ that $n$ vertices, $x^a_1,\ldots,x^a_n$, of a lattice with lattice spacing $a$ belong to the same cluster (precise definitions are given in Section \ref{sec:proofs} below). According to \cite{JS19}, obtaining closed-form expressions for such objects is considered a ``holy grail'' in the field.

This CFT approach was discussed by Aizenman in a talk presented at the 12th International Congress on Mathematical Physics (ICMP 97, see \cite{Aizenman98A}) and in Section 13 of \cite{Aizenman98B}, where Theorems \ref{thm:scal-lim-connection-probabilities} and \ref{thm:scal-lim-bounded-domains} of the present work are essentially conjectured. It was also discussed by Schramm and Smirnov in \cite{SS11}, where it is pointed out that it would be natural to study the scaling limit of connection probabilities in conjunction with the development of a corresponding CFT. Indeed, the first step towards a mathematical theory of such a percolation CFT is a proof of conformal invariance of connection probabilities together with the identification of a field whose correlation functions can be expressed in terms of those probabilities. In this article we solve both problems by proving the conjectures presented in Section 13 of \cite{Aizenman98B} (see Theorems \ref{thm:scal-lim-connection-probabilities} and \ref{thm:scal-lim-bounded-domains} below) and by identifying an appropriate conformal field (see Theorems \ref{thm:scal-lim-correlation-functions} and \ref{thm:Sobolev-conv-field}, and Corollary~\ref{cor:conf-cov}).

These results have immediate consequences of interest. For example, consider the quantity $P^a_3(x^a_1,x^a_2,x^a_3)/\sqrt{P^a_2(x^a_1,x^a_2)P^a_2(x^a_1,x^a_3)P^a_2(x^a_2,x^a_3)}$, which has attracted significant attention, in the physics and the mathematics literature, and was conjectured to converge, in the scaling limit $a \to 0$, to a universal constant $R$. 
For percolation in the upper half-plane (or any domain conformally equivalent to the upper half-plane), this ratio was considered in \cite{KSZ06,SKZ07}, where it was argued, both theoretically and numerically, that it should tend to a constant as $a \to 0$.
A proof of this result directly in the scaling limit, including an explicit expression for the value of $R$, was obtained in \cite{BI12} using SLE$_6$ calculations. The corresponding result on the triangular lattice, again in the upper half-plane, was obtained in \cite{Conijn15}.

Delfino and Viti \cite{DV11} studied the same ratio on the plane and derived a universal expression for $R$ using conformal field theory methods and connections between percolation and Potts models. This expression can be computed analytically and gives $R \approx 1.022$, which is in agreement with the numerical value found previously in a study of $R$ on the cylinder \cite{SKZ09}. A numerical verification of this value for the plane was obtained in \cite{ZSK11}, while a proof that the ratio tends to a constant as $a \to 0$ follows immediately from our Theorem \ref{thm:scal-lim-connection-probabilities} (see Corollary \ref{cor:ratio}).

To the best of our knowledge, there is no rigorous derivation of the numerical value of $R$ on the plane, but we point out that a proof of the Delfino-Viti formula for $R$ directly in the scaling limit, in the context of conformal loop ensembles, may be forthcoming (see Section 1.4 of \cite{AS21}). Obtaining rigorous results and exact formulas on the plane is often more challenging than on the upper half-plane or in finite domains because SLE techniques are not so readily applicable, due to the absence of a boundary.

In CFT, the behavior expressed by Theorems \ref{thm:scal-lim-connection-probabilities} and \ref{thm:scal-lim-bounded-domains} below is usually associated with the correlation functions of \emph{primary fields} (see, e.g., \cite{DFMS,Henkel}). It is then natural to ask if one can identify a lattice field whose correlation functions converge to the scaling limit of connection probabilities. In the physics literature, reference to a percolation field is often avoided by deriving results for the $q$-state Potts model and then extrapolating them to percolation by taking the limit $q \to 1$. This is the case, for example, in \cite{Cardy92} and \cite{DV11}, as well as in \cite{KSZ06,SKZ07}, which rely on \cite{Cardy92}. Working with the Potts model, which has a well-defined magnetization field, makes it possible to use conformal field theory tools that are not directly available for percolation. However, the limit $q \to 1$ hides a remarkable amount of subtlety (see, e.g., \cite{CR13}) and cannot be justified rigorously. In this respect, the identification and study of a percolation field (see Theorems \ref{thm:scal-lim-correlation-functions} and \ref{thm:Sobolev-conv-field} and Corollary \ref{cor:conf-cov}) could help formulate conformal field theory results for percolation more directly, without relying on the use of the $q$-state Potts model and the extrapolation $q \to 1$. We provide some evidence of this in Section \ref{sec:4-point-function}, where we show that the analysis of the four-point function of the percolation field we introduce in this paper suggests the presence of an additional field of scaling dimension $5/4$, which can be identified with the so-called ``four-leg operator'' (see \cite{VJS12} and references therein).

\subsection{Definitions and main results} \label{sec:main-results}

We consider critical site percolation on $a\mathcal{T}$, the triangular lattice $\mathcal{T}$ scaled by a factor $a>0$. We embed $\mathcal{T}$ in $\mathbb{R}^2$ as in Figure~\ref{fig-tri-hex} and in such a way that one of its vertices coincides with the origin of $\mathbb{R}^2$. We denote this vertex by $0$ and call it the \emph{origin} of $\mathcal{T}$. Each vertex of $a\mathcal{T}$ is identified with the elementary cell of $a\mathcal{H}$ of which it is the center, where $\mathcal{H}$ is the hexagonal lattice dual to $\mathcal{T}$ (see Figure \ref{fig-tri-hex}). Each vertex of $\mathcal{T}$ (or hexagonal cell of $\mathcal{H}$) is declared \emph{open} or \emph{closed} with equal probability, independently of all other vertices. With probability one, all open and closed clusters (maximal connected components of the sets of open and closed vertices, respectively) are finite (for the percolation model considered here, this follows from the methods of \cite{Harris60}), so the boundaries between open and closed clusters can be represented as loops drawn using edges of the dual hexagonal lattice. We call these loops (percolation) \emph{interfaces}.

\begin{figure}[!ht]
	\begin{center}
		\includegraphics[width=5cm]{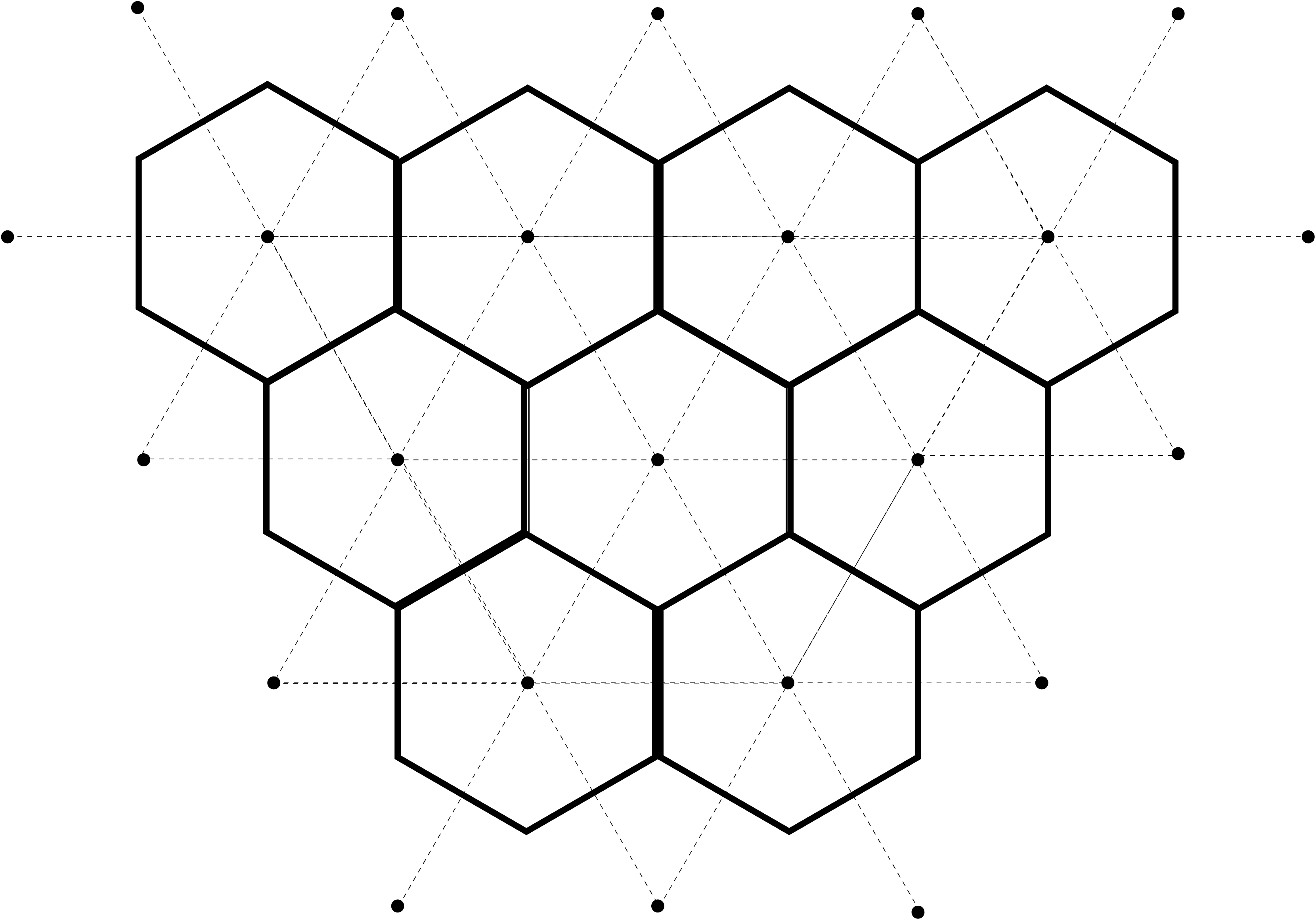}
		\caption{Embedding of the triangular and hexagonal lattices in ${\mathbb R}^2$.}
		\label{fig-tri-hex}
	\end{center}
\end{figure}

A special role will be played by the probability that the open cluster of the origin reaches the circle of radius one, which we will denote by $\pi_a$.

For any collection of vertices $x^a_1,\ldots,x^a_n \in a\mathcal{T}$, we let $P^a_{n}(x^a_1,\ldots,x^a_n)$ denote the probability that $x^a_1,\ldots,x^a_n$ belong to the same open cluster. Note that the probability that $x^a_1,\ldots,x^a_n$ belong to the same cluster, either open or closed, is simply $2P^a_{n}(x^a_1,\ldots,x^a_n)$.

Our first result concerns the scaling limit of such connection probabilities as $a \to 0$. In the following theorem and in the rest of the paper, we identify $\mathbb{R}^2$ with the complex plane~$\mathbb{C}$, so $x,x_i$ will denote both elements of $\mathbb{R}^2$ and complex numbers. When we refer to M\"obius transformations $M(x)=\frac{ax+b}{cx+d}$ (or to more general conformal maps), $x$ is interpreted as a complex number and $M$ is always assumed to be non-singular, that is, we assume that $ad-bc \neq 0$.

\begin{theorem} \label{thm:scal-lim-connection-probabilities}
	For any $n \geq 2$ and any collection of distinct points $x_1,\ldots,x_n$ in $\mathbb{C}$, let $ x^a_1,\ldots,x^a_n \in a\mathcal{T}$ be chosen in such a way that $x^a_i \to x_i$, as $a \to 0$, for each $i=1,\ldots,n$. Then the following limit exists and is nontrivial:
	\begin{align} \label{eq:scal-lim-connection-probabilities}
	& P_{n}(x_1,\ldots,x_n) := \lim_{a \to 0} \pi_a^{-n} P^a_{n}(x^a_1,\ldots,x^a_n).
	\end{align}
	Moreover, if $M$ is a non-singular M\"obius transformation mapping $x_i \mapsto M(x_i) \in \mathbb{C}$ for each $i=1,\ldots,n$,
	\begin{equation} \label{eq:conf-cov-Pn}
	P_{n}(M(x_1),\ldots,M(x_n)) = \Big( \prod_{i=1}^n \vert M'(x_i) \vert^{-5/48} \Big) P_{n}(x_1,\ldots,x_n).
	\end{equation}
\end{theorem}

The theorem immediately implies that there is a constant $C_2$ such that 
\begin{align} \label{eq:2-point-function}
P_2(x_1,x_2) = C_2 \Vert x_1-x_2 \Vert^{-5/24}.
\end{align}
(This result follows also from Proposition 5.3 of \cite{GPS13}.)
Moreover, by standard arguments (see, for example, \cite{DFMS} or the proof of Theorem 4.5 of \cite{CGK16}), Theorem \ref{thm:scal-lim-connection-probabilities} also implies that there is a constant $C_3$ such that
\begin{align} \label{eq:3-point-function}
P_3(x_1,x_2,x_3) = C_3 \Vert x_1-x_2 \Vert^{-5/48} \Vert x_1-x_3 \Vert^{-5/48} \Vert x_2-x_3 \Vert^{-5/48}.
\end{align}
This has the following immediate consequence.
\begin{corollary} \label{cor:ratio}
	The ratio
	\begin{align}
	\frac{P_3(x_1,x_2,x_3)}{\sqrt{P_2(x_1,x_2)P_2(x_1,x_3)P_2(x_2,x_3)}} = \frac{C_3}{C_2^{3/2}}
	\end{align}
	is independent of $x_1,x_2,x_3$.
\end{corollary}

The proof of Theorem \ref{thm:scal-lim-connection-probabilities}, presented in Section \ref{sec:proofs}, makes essential use of the full scaling limit of critical percolation constructed in \cite{CN06} and of its invariance properties under M\"obius transformations \cite{GMQ21}.

\begin{remark} \label{rem:4-point-function}
	The proof of Theorem \ref{thm:scal-lim-connection-probabilities} can be adapted to deal with other connectivity functions. Consider, e.g., the probability $P^a_4(x^a_1,x^a_2 \vert x^a_3,x^a_4)$ that $x^a_1,x^a_2$ belong to the same open cluster and $x^a_3,x^a_4$ belong to a different open cluster. Probabilities of this type play a crucial role in the percolation CFT \cite{DV11bis,PRS16,JS19}. Arguments analogous to those used in the proof of Theorem \ref{thm:scal-lim-connection-probabilities} show that $\pi_a^{-4} P^a(x^a_1,x^a_2 \vert x^a_3,x^a_4)$ has a conformally covariant scaling limit as $a \to 0$. This can be understood observing that the difference between $P^a_4(x^a_1,x^a_2 \vert x^a_3,x^a_4)$ and $P^a_4(x^a_1,x^a_2, x^a_3,x^a_4)$ consists in the presence of a percolation interface separating $x^a_1,x^a_2$ from $x^a_3,x^a_4$, an event that can be dealt with adapting the methods used in the proof of Theorem \ref{thm:scal-lim-connection-probabilities} (see the proof of Theorem \ref{thm:scal-lim-correlation-functions} in Section \ref{sec:proofs}).
\end{remark}

A result analogous to Theorem \ref{thm:scal-lim-connection-probabilities} is valid for the scaling limit of connection probabilities in bounded domains and in any domain equivalent to the upper half-plane. In order to state this result, we let $P^a_{D,n}(x^a_1,\ldots,x^a_n)$ denote the probability that $x^a_1,\ldots,x^a_n$ are in the same open cluster for a percolation model on $a\mathcal{T}$ such that all vertices outside $a\mathcal{T} \cap D$ are declared closed. Unless otherwise stated, in the rest of the paper, when we discuss conformal maps, a domain $D$ will be an open subset of $\mathbb{C}$.

\begin{theorem} \label{thm:scal-lim-bounded-domains}
	Let $D$ be a domain conformally equivalent to the upper half-plane. For any $n \geq 2$ and any collection of distinct points $x_1,\ldots,x_n \in D$, let $ x^a_1,\ldots,x^a_n \in a\mathcal{T}$ be chosen in such a way that $x^a_i \to x_i$, as $a \to 0$, for each $i=1,\ldots,n$. Then the following limit exists and is nontrivial:
	\begin{align} \label{eq:scal-lim-bounded-domain}
	& P_{D,n}(x_1,\ldots,x_n) := \lim_{a \to 0} \pi_a^{-n} P^a_{D,n}(x^a_1,\ldots,x^a_n).
	\end{align}
	Moreover, if $\phi:D \to D'$ is a conformal map from $D$ to $D'$, then
	\begin{equation} \label{eq:conf-cov-bounded-domains}
	P_{D',n}(\phi(x_1),\ldots,\phi(x_n)) = \Big( \prod_{i=1}^n \vert \phi'(x_i) \vert^{-5/48} \Big) P_{D,n}(x_1,\ldots,x_n).
	\end{equation}
\end{theorem}

The proof of this theorem is similar to that of Theorem \ref{thm:scal-lim-connection-probabilities} and is presented in Section~\ref{sec:proofs}.

As mentioned earlier, in CFT the behavior in Theorems \ref{thm:scal-lim-connection-probabilities} and \ref{thm:scal-lim-bounded-domains} is associated with the correlation functions of \emph{primary fields} (see, e.g., \cite{DFMS,Henkel}). In the case of percolation, letting $\mathbf{1}_{\{\cdot\}}$ denote the indicator function, one can naively consider the natural (centered) lattice field $\omega(x^a) = \mathbf{1}_{\{x^a \text{ is open}\}} - 1/2$, whose ``integral'' captures the spatial fluctuations of the density of open vertices; or some version $\omega_{\varepsilon}(x^a)$ where the event that $x^a$ is open is replaced by the event that its open cluster reaches the circle of radius $\varepsilon$ centered at $x^a$, for some $\varepsilon>0$. However, the random variables $\omega_{\varepsilon}(x^a_1)$ and $\omega_{\varepsilon}(x^a_2)$ are clearly independent when the distance between $x^a_1$ and $x^a_2$ is greater than $2\varepsilon$, so the correlation functions of this lattice field are identically zero at large distances and have a trivial scaling limit. In what follows we introduce a spin model whose correlation functions capture some of the percolation connectivity properties.

Let $\{\mathcal{C}^a_i\}_i$ denote the collection of open clusters on $a\mathcal{T}$ and assign to each cluster $\mathcal{C}^a_i$ a random sign $\sigma_i=\pm 1$, where $\{\sigma_i\}_i$ is a collection of symmetric, $(\pm 1)$-valued, i.i.d.\ random variables. For each $x^a \in a\mathcal{T}$, we let 
\begin{align} \label{def:spin-field}
S_{x^a}=\begin{cases}
\sigma_i \;\; \text{ if } x^a \in \mathcal{C}^a_i \\
0 \;\;\; \text{ if $x^a$ is closed }
\end{cases}
\end{align}
Models of this type are called ``\emph{divide and color}'' and were studied, for example, in \cite{Haggstrom01,BCM09,Balint-thesis,Balint10,BBT13,ST19}.

In our context, it is natural to introduce the lattice field $\mathbf{S} := \{ S_{x^a} \}_{x^a \in a\mathcal{T}}$ because, as we will show below, its nonvanishing correlation functions can be expressed in terms of the connection probabilities $P^a_{2n}$ in such a way that they have a well-defined scaling limit with the same covariance properties as those of $P_{2n}$. To see this, for a collection $x^a_1,\ldots,x^a_{2n}$ of distinct vertices of $a\mathcal{T}$, let $\mathcal{Q}(x_1,\ldots,x_{2n})$ denote the set of all partitions $(Q_1,\ldots,Q_k)$ of $x^a_1,\ldots,x^a_{2n}$ such that each element $Q_j$ contains an even number of vertices. Moreover, for a partition $(Q_1,\ldots,Q_k) \in \mathcal{Q}(x_1,\ldots,x_{2n})$, let $G(Q_1,\ldots,Q_k)$ denote the event that all vertices in the same element of the partition belong to the same open cluster and no two vertices in different elements of the partition belong to the same cluster.

Then, if we let $\langle \cdot \rangle^a$ denote the expectation with respect to the distribution of open clusters, $P^a$, and of the $\sigma_i$'s, for any collection $x^a_1,\ldots,x^a_{n}$ of distinct vertices of $a\mathcal{T}$, we have
\begin{align} \label{eq:n-point-function}
\langle S_{x_1^a} \ldots S_{x_{n}^a} \rangle^a = \begin{cases}
\sum_{(Q_1,\ldots,Q_k)\in\mathcal{Q}} P^a(G(Q_1,\ldots,Q_k)) \; \text{ if $n$ is even} \\
0 \;\;\;\;\;\;\;\;\;\;\;\;\;\;\;\;\;\;\;\;\;\;\;\;\;\;\;\;\;\;\;\;\;\;\;\;\;\;\;\;\;\;\;\;\;\;\;\;\; \text{ if $n$ is odd}
\end{cases}
\end{align}
This can be seen by taking the expectation in two steps, first conditioning on the percolation configuration, which determines the clusters $\{ \mathcal{C}^a_i \}_i$, and averaging over the signs $\{ \sigma_i \}_i$, and then summing over all possible percolation configurations. The independence of the $\sigma_i$'s implies that, if an odd number of the vertices $x_1,\ldots,x_n$ is contained in an open cluster $\mathcal{C}^a_i$, the expectation vanishes by symmetry. Note that, in particular, $\langle S_{x_1^a} S_{x_{2}^a} \rangle^a  = P^a_2(x^a_1,x^a_2)$.

The field \eqref{def:spin-field} can also be defined in a domain $D$ that is not the whole plane by declaring closed all vertices outside $a\mathcal{T} \cap D$. In this case, we will denote the expectation defined above by~$\langle \cdot \rangle^a_{D}$.

\begin{theorem} \label{thm:scal-lim-correlation-functions}
	Let $D$ be a simply-connected domain of $\mathbb{C}$ (possibly $\mathbb{C}$ itsef).
	For any $n \geq 2$ and any collection of distinct points $x_1,\ldots,x_n \in D$, let $ x^a_1,\ldots,x^a_n \in a\mathcal{T}$ be chosen in such a way that $x^a_i \to x_i$, as $a \to 0$, for each $i=1,\ldots,n$. Then the following limit exists and is nontrivial:
	\begin{equation} \label{eq:scal-lim-correlation-functions}
	\mathrm{C}_{D,n}(x_1,\ldots,x_n) := \lim_{a \to 0} \pi_a^{-n} \langle S_{x_1^a} \ldots S_{x_{n}^a} \rangle^a_D.
	\end{equation}
	Moreover, if $\phi:D \to D'$ is a conformal map from $D$ to $D'$, then
	\begin{equation} \label{eq:conf-cov-Pn}
	\mathrm{C}_{D',n}(\phi(x_1),\ldots,\phi(x_n)) = \Big( \prod_{i=1}^n \vert \phi'(x_i) \vert^{-5/48} \Big) \mathrm{C}_{D,n}(x_1,\ldots,x_n).
	\end{equation}	
\end{theorem}

The proof of this theorem, presented in Section \ref{sec:proofs}, follows those of Theorems \ref{thm:scal-lim-connection-probabilities} and \ref{thm:scal-lim-bounded-domains}, but it requires some modifications due to the different type of events involved in \eqref{eq:n-point-function}, as mentioned in Remark \ref{rem:4-point-function}.

As observed above, the $n$-point functions of the lattice field $\mathbf{S}$ are identically zero when $n$ is odd, so $\mathrm{C}_{2n+1} \equiv 0$ for all $n$. More generally, the functions $\mathrm{C}_{n}$ are not related to the functions $P_n$ for $n$ odd. This is analogous to what happens in the case of the Ising model between the $n$-point functions of the spin field and the connection probabilities of the Fortuin-Kasteleyn percolation model with $q=2$. A discussion of this phenomenon from the perspective of conformal field theory is contained in Section 5.2 of \cite{PSVD13}.

\begin{remark}
	The connectivity properties of both open and closed clusters together can be captured by a version, $\tilde{\mathbf{S}}$, of the field $\mathbf{S}$ in which random signs are associated to both open and closed clusters. The field $\mathbf{S}$ is a site-diluted version of $\tilde{\mathbf{S}}$. Because of the symmetry between open and closed clusters, the results we present below for $\mathbf{S}$ and the lattice fields derived from it are also valid for $\tilde{\mathbf{S}}$ and the corresponding fields.
\end{remark}

According to Theorem \ref{thm:scal-lim-correlation-functions}, the correlation functions of $\mathbf{S}$ have a well-defined scaling limit; it is therefore natural to ask if the field itself has a well-defined scaling limit. To answer this question, we introduce the lattice field
\begin{equation}\label{def:lattice-field}
\Phi^{a} := a^2 \pi_a^{-1} \sum_{x^a \in a\mathcal{T}} S_{x^a} \delta_{x^a},
\end{equation}
where $\delta_{x^a}$ is a unit Dirac point measure at $x^a$. More precisely, for functions $f$ of bounded support on $\mathbb{R}^2$, we define
\begin{align} \label{eq:lat-field1}
\begin{split}
& \Phi^{a}(f) \equiv \int_{{\mathbb R}^2} f(x) \Phi^{a}(x) dx := \int_{{\mathbb R}^2} f(z) \big[ a^2 \pi_a^{-1} \sum_{x^a \in a\mathcal{T}} S_{x^a} \delta(x-x^a) \big] dx \\
& \qquad \quad = a^2 \pi_a^{-1} \sum_{x^a \in a\mathcal{T}} f(x^a) S_{x^a} = a^2 \pi_a^{-1} \sum_i \sigma_i \sum_{x^a \in \mathcal{C}^a_i} f(x^a),
\end{split}
\end{align}
where the open clusters $\mathcal{C}^a_i$ are seen as subsets of the vertices of $a\mathcal{T}$.

In particular, if $\mathbf{1}_{D}$ denotes the indicator function of a finite domain $D$,
\begin{align}
& \Phi^{a}(\mathbf{1}_{D}) = a^{2} \pi_a^{-1} \sum_{x^a \in a\mathcal{T} \cap D} S_{x^a} = a^{2} \pi_a^{-1} \sum_{i: \mathcal{C}^a_i \in \mathscr{C}_D^a} \sigma_i \, \vert \mathcal{C}^a_i \cap D \vert,
\end{align}
where $\mathscr{C}_D^a$ is the collection of open clusters that intersect $D$ and $\vert \mathcal{C}^a_i \cap D\vert$ is the number of vertices in $\mathcal{C}^a_i \cap D$.

More generally, if we introduce the normalized counting measures
\begin{align} \label{def:counting-measure}
\mu^a_i := a^{2} \pi_a^{-1} \sum_{x^a \in \mathcal{C}^a_i} \delta_{x^a},
\end{align}
we can write
\begin{align} \label{eq:lattice-field}
\Phi^{a}(f) = \sum_i \sigma_i \, \mu^a_i(f).
\end{align}

This way of expressing the field is useful because it was shown in \cite{CCK19} that, as $a \to 0$, the collection of normalized counting measures $\{ \mu_i^a \}_i$ converges in distribution (in an appropriate topology) to a collection of finite measures $\{ \mu_k \}_k$, which is measurable with respect to the continuum scaling limit of percolation in terms of interfaces constructed in \cite{CN06}. 

Using \eqref{eq:n-point-function} and \eqref{eq:2-point-function}, it is easy to see that, for any bounded function $f$ of bounded support, the random variable $\Phi^a(f)$ is tight and therefore has subsequential limits in distribution as $a \to 0$. To obtain stronger results, it is convenient to work with a ``smoother'' version of the field $\Phi^a$. For this reason, we let $D_n \equiv [-n,n]^2$ and introduce the functions
\begin{align} \label{def:lattice-smooth-field-full-plane}
\hat{\Phi}^a(x) := a^2 \pi_a^{-1} \sum_{i} \sigma_i \sum_{x^a \in \mathcal{C}^a_i} \frac{\mathbf{1}_{x^a}(x)}{A_a}
\end{align}
and
\begin{align} \label{def:lattice-smooth-field}
\hat{\Phi}^a_{n}(x) := a^2 \pi_a^{-1} \sum_{i:\mathcal{C}^a_i \in \mathscr{C}^a_{D_n}} \sigma_i \sum_{x^a \in \mathcal{C}^a_i \cap D_n} \frac{\mathbf{1}_{x^a}(x)}{A_a},
\end{align}
where $\mathbf{1}_{x^a}$ is the indicator function of the elementary hexagon $x_a$ of $a\mathcal{H}$ (the hexagon centered at $x_a \in a\mathcal{T}$---recall that, with a slight abuse of notation, we use $x^a$ both for elementary hexagons of $a\mathcal{H}$ and their centers in $a\mathcal{T}$) and $A_a$ denotes the area of an elementary hexagon of $a \mathcal{H}$.

With these definitions we have the following results (the definitions of the Sobolev spaces $H^1$ and $H^{-1}$ and of the norm $\Vert \cdot \Vert_{H^{-1}}$ are given in Section \ref{sec:field-s-lim} below).

\begin{theorem} \label{thm:Sobolev-conv-field}
	The lattice field $\hat\Phi^{a}$ has a unique scaling limit in the following sense. There exists a random element $\Phi$ of the Sobolev space $H^{-1}$ such that, for any $n \in \mathbb{N}$, as $a \to 0$, $\hat\Phi_n^{a}$ converges in distribution to $\Phi \vert_{[-n,n]^2}$, the restriction of $\Phi$ to $[-n,n]^2$ (more precisely, to functions in $H^1([-n,n]^2)$). The convergence is in the topology induced by the norm $\Vert \cdot \Vert_{H^{-1}}$.
	
	Moreover, $\Phi$ can be approximated using the collection of limiting measures $\{ \mu_k \}_k$ and i.i.d.\ symmetric random signs $\sigma_k=\pm 1$ in the sense that, for any smooth function $f$ of bounded support, there is a coupling of $\Phi(f)$ and $\{ \mu_k(f) \}_k$ such that, if $\langle \cdot \rangle$ denotes expectation,
	\begin{align}
	\Big\langle \Big\vert \Phi(f) - \sum_{k:\diam(\supp(\mu_k))>\varepsilon} \sigma_k \, \mu_k(f) \Big\vert^2 \Big\rangle = O(\varepsilon^{43/24}) \;\;\; \text{as } \varepsilon \to 0.
	\end{align}
\end{theorem}

\begin{corollary} \label{cor:conf-cov}
	The field $\Phi$ of Theorem \ref{thm:Sobolev-conv-field} is translation and rotation invariant and is scale covariant in the sense that, formally, for any $s>0$, $\Phi(sx)=s^{-5/48}\Phi(x)$. More precisely, the field $\Phi_s$ defined by 
	\begin{align}
	\Phi_s(f) := \int_{{\mathbb R}^2} f\Big(\frac{x}{s}\Big) \Phi(x) dx
	\end{align}
	has the same distribution as $s^{2-5/48}\Phi$.
	In particular, for any $L,s>0$, the distribution of $\Phi(\mathbf{1}_{[-sL,sL]^2})$ is the same as that of $s^{91/48}\Phi(\mathbf{1}_{[-L,L]^2})$.
\end{corollary}

The proofs of Theorem \ref{thm:Sobolev-conv-field} and Corollary \ref{cor:conf-cov}, presented in Section \ref{sec:field-s-lim}, use ideas and tools from \cite{Dubedat09,CGN15,vdBCL18,CGPR21,CJN21}, as well as properties of the conformal measure ensemble for critical percolation whose existence was conjectured in \cite{CN09} and proved in \cite{CCK19}.

\medskip

\section{Scaling limit of connection probabilities and correlation functions} \label{sec:proofs}

\subsection{Percolation interfaces and their scaling limit} \label{sec:interfaces}

Before we can present the proofs of Theorems \ref{thm:scal-lim-connection-probabilities} and \ref{thm:scal-lim-correlation-functions}, we need to introduce some additional 
notation and results.

We let $P^a$ denote the probability distribution of critical percolation on $a\mathcal{T}$ (i.e., a Bernoulli product measure corresponding to the assignment of a label---open or closed--- to each vertex of $a\mathcal{T}$ with equal probability) and will use $\Lambda^a$ to denote a percolation configuration in $a\mathcal{T}$ distributed according to $P^a$ or, equivalently, the corresponding configuration of percolation interfaces. The percolation interfaces between open and closed clusters can be given a direction and seen as oriented curves, with the direction depending on whether an interface surrounds an open or a closed cluster.

More precisely, for each fixed $a>0$, the percolation interfaces are polygonal circuits (with probability one) on the edges of the hexagonal lattice $a\mathcal{H}$ dual to the triangular lattice $a\mathcal{T}$. We give these circuits an orientation in such a way that they wind counterclockwise around open clusters and clockwise around closed clusters (in other words, they are oriented in such a way that open hexagons are on the left and closed hexagons on the right). Note that the interfaces form a nested collection of loops with alternating orientation and a natural tree structure.

We let $B_{r}(x)$ denote the disk of radius $r$ centered at $x$ and $\partial B_{r}(x)$ the circle of radius $r$ centered at $x$. We write $A_{r,R}(x) := B_{R}(x) \setminus B_{r}(x)$ for the annulus centered at $x$ with inner radius $r$ and outer radius $R$.
We write $x_1^a \longleftrightarrow x_2^a$ to denote the event that $x^a_1$ and $x^a_2$ are in the same open cluster, which implies that there is an \emph{open path} from $x_1^a$ to $x_2^a$, that is, a sequence of nearest-neighbor open vertices of $a\mathcal{T}$ starting at $x_1^a$ and ending at~$x_2^a$. 

We also let $\mathcal{A}^a_{0,r}(x^a) \equiv x^a \longleftrightarrow \partial B_{r}(x^a)$ denote the event that the open cluster of $x^a$ (as a set of hexagons) intersects $\partial B_{r}(x^a)$, and $\mathcal{A}^a_{r,R}(x^a) \equiv \partial B_{r}(x^a) \longleftrightarrow \partial B_{R}(x^a)$ denote the event that there is an open path that crosses the annulus $A_{r,R}(x^a)$ in the sense that it starts inside $B_{r}(x^a)$ and exits $B_{R}(x^a)$. This implies that there is no interface loop between open and closed hexagons that separates $B_{r}(x^a)$ from the complement of $B_{R}(x^a)$ (i.e., is contained in $A_{r,R}(x^a)$ and surrounds $B_{r}(x^a)$) and that, in addition, one of the two following events occurs (with $P^a$-probability one):
\begin{itemize}
	\item there is a counterclockwise interface intersecting both $B_{r}(x^a)$ and the complement of $B_{R}(x^a)$,
	\item the innermost interface surrounding $B_{r}(x^a)$ is oriented counterclockwise.
\end{itemize}

If $x^a_1,x^a_2 \in a\mathcal{T}$, for any $r_1,r_2>0$, $B_{r_1}(x^a_1) \longleftrightarrow B_{r_2}(x^a_2)$ will denote the event that there is an open path that starts inside $B_{r_1}(x^a_1)$ and ends inside $B_{r_2}(x^a_2)$. In terms of interfaces, this means that there is no interface separating $B_{r_1}(x^a_1)$ from $B_{r_2}(x^a_2)$ and that, in addition, one of the following events occurs (with $P^a$-probability one):
\begin{itemize}
	\item there is a counterclockwise interface intersecting both $B_{r_1}(x^a_1)$ and $B_{r_2}(x^a_2)$,
	\item there is a counterclockwise interface surrounding either $B_{r_1}(x^a_1)$ or $B_{r_2}(x^a_2)$ and intersecting the other,
	\item the innermost interface surrounding both $B_{r_1}(x^a_1)$ and $B_{r_2}(x^a_2)$ is oriented counterclockwise.
\end{itemize}

The probability of the \emph{one-arm event} $\mathcal{A}^a_{0,1} = 0 \longleftrightarrow \partial B_{1}(0)$ will play a special role and, for each $a>0$, will be denoted by $\pi_a \equiv P^a(0 \longleftrightarrow \partial B_{1}(0))$.
It follows from the results of \cite{LSW02,GPS13} (see the discussion in Section~5.1 of \cite{GPS13}, in particular the first limit in the third displayed equation on page 999) that, for any $\varepsilon>0$,
\begin{align}
\begin{split} \label{eq:one-arm-limit}
\lim_{a \to 0} \pi_a^{-1} P^a\big(0 \longleftrightarrow \partial B_{\varepsilon}(0)\big) = \varepsilon^{-5/48}.
\end{split}
\end{align}

In the scaling limit, the collection of interfaces between open and closed clusters converges weakly to a random ensemble of fractal nonsimple loops in the topology generated by the distance function $\text{Dist}$ introduced below.

To define $\text{Dist}$, we first introduce a function $\Delta$ on ${\mathbb R}^2 \times {\mathbb R}^2$ give by 
\begin{equation} \label{new-dist}
\Delta(u,v) := \inf_{\varphi} \int_0^1\frac{\Vert\varphi'(t)\Vert}{1+\Vert\varphi(t)\Vert^2}dt,
\end{equation}
where the infimum is over all differentiable curves $\varphi:[0,1] \to \mathbb{R}^2$ with $\varphi(0)=u$ and $\varphi(1)=v$. The function $\Delta$ induces a metric equivalent to the Euclidean metric in bounded regions, but it has the advantage of making ${\mathbb R}^2$ precompact. Adding a single point at infinity yields the compact space
$\dot{\mathbb R}^2$, which is isometric, via stereographic projection, to the two-dimensional sphere.

We then define a distance between two planar loops, $\gamma_1, \gamma_2:[0,1] \rightarrow \mathbb{R}^2$, seen as oriented curves, as follows:
\begin{equation} \label{def:d}
\text{d}(\gamma_1,\gamma_2) := \inf\sup_{t \in [0,1]} \Delta(\gamma_1(t) - \gamma_2(t)),
\end{equation}
where the infimum is over all choices of parametrizations (with the same orientations) of $\gamma_1$ and $\gamma_2$. 

Finally, we define a distance between two closed sets of loops, $\Gamma_1$ and $\Gamma_2$, as follows:
\begin{align} \label{def:Dist}
& \text{Dist}(\Gamma_1,\Gamma_2) \\
& \qquad := \inf\{\varepsilon>0: \forall \gamma_1\in \Gamma_1~\exists \gamma_2\in \Gamma_2 \text{ s.t. } \text{d}(\gamma_1,\gamma_2)\leq\varepsilon\text{ and vice versa}\}. \nonumber
\end{align}
The space $\Omega$ of collections of loops with this distance is a separable metric space. 

It was shown in \cite{CN06} that, as $a \to 0$, the collection of percolation interfaces has a unique limit in distribution in the topology induced by \eqref{def:Dist}. We call this limit the \emph{full scaling limit} of percolation and let $\mathbb{P}$ denote its distribution. A loop configuration distributed according to $\mathbb{P}$ will be denoted by $\Lambda$. As explained in \cite{CN08}, $\Lambda$ is distributed like the full-plane, nested conformal loop ensemble CLE$_6$. It is invariant, in a distributional sense, under all M\"obius transformations \cite{CN06,GMQ21}.

For the scaling limit, we will use notation similar to that introduced above for discrete percolation. In particular, $\mathcal{A}_{r,R}(x) \equiv \partial B_{r}(x) \longleftrightarrow \partial B_{R}(x)$ and $B_{r_1}(x_1) \longleftrightarrow B_{r_2}(x_2)$ are the events described in the discrete setting in terms of loops. Indeed, the definitions of the corresponding lattice events in terms of loops make sense in the continuum as well as on the lattice, provided we use the following additional definition. We say that an interface $\gamma$ \emph{separates} two domains, $D_1$ and $D_2$, if there are no points $y_1 \in D_1$ and $y_2 \in D_2$ with the same winding number with respect to $\gamma$; otherwise, we say that $\gamma$ does not separate $D_1$ and $D_2$.
An equivalent way to say this is to declare the \emph{outside} of a curve $\gamma$ to be the set of points of $\mathbb{R}^2 \setminus \gamma$ whose winding number with respect to $\gamma$ is zero, and the \emph{inside} of $\gamma$ to be the complement of that set in $\mathbb{R}^2 \setminus \gamma$. Then $y_1,y_2 \notin \gamma$ are not separated by $\gamma$ if they are both in its inside or both in its outside. This is the correct notion to ensure continuity of connectivity events in the scaling limit because, although the lattice interfaces are simple (self-avoiding) curves, critical percolation clusters have deep ``fjords'' and, in the scaling limit, the curves $\gamma$ are non-intersecting but not simple (they are self-touching).
We point out that, alternatively, one could define the above events in terms of the ensemble of continuum clusters constructed in \cite{CCK19}.


For $\Vert x_1-x_2 \Vert>r_1+r_2$, the boundary of the event $B_{r_1}(x_1) \longleftrightarrow B_{r_2}(x_2)$ is the event of $\mathbb{P}$-probability zero that a loop of diameter at least $\min\big(\Vert x_1-x_2 \Vert - (r_1+r_2), r_1, r_2 \big)$ touches either $\partial B_{r_1}(x_1)$ or $\partial B_{r_2}(x_2)$ without crossing it, so that $B_{r_1}(x_1) \longleftrightarrow B_{r_2}(x_2)$ is a continuity event for $\mathbb{P}$. Similar considerations apply to the event $\mathcal{A}_{r,R}(x) \equiv \partial B_{r}(x) \longleftrightarrow \partial B_{R}(x)$.
Therefore, for any $R>r>0$, any $x,x_1,x_2 \in \mathbb{R}^2$, and any sequences $x^a, x_1^a, x_2^a \to x, x_1, x_2$, respectively, as $a \to 0$, the convergence of critical percolation interfaces to their scaling limit in the topology induced by \eqref{def:Dist} (equivalently, the convergence of critical percolation clusters \cite{CCK19}) implies that
\begin{align} \label{eq:conv-arm-prob}
\lim_{a \to 0} P^a(\partial B_{r}(x^a) \longleftrightarrow \partial B_{R}(x^a)) = \mathbb{P}(\partial B_{r}(x) \longleftrightarrow \partial B_{R}(x))
\end{align}
and
\begin{align} \label{eq:conv-disks-prob}
\lim_{a \to 0} P^a(B_{r_1}(x^a_1) \longleftrightarrow B_{r_2}(x^a_2)) = \mathbb{P}(B_{r_1}(x_1) \longleftrightarrow B_{r_2}(x_2)).
\end{align}

\medskip

\subsection{Proofs of Theorems \ref{thm:scal-lim-connection-probabilities}, \ref{thm:scal-lim-bounded-domains} and \ref{thm:scal-lim-correlation-functions}}

We start this section with a lemma that will be used in the proof of Theorem \ref{thm:scal-lim-connection-probabilities}. The lemma concerns discrete percolation for fixed lattice spacing $a>0$, but it plays an important role in the proof of Theorem \ref{thm:scal-lim-connection-probabilities}, so before stating and proving it, we briefly explain what it says and how it is used later on in the paper. Consider three scales, $\varepsilon$, $\delta$ and $\eta$, such that  $\varepsilon >\delta \gg \eta > a$. For fixed $a$, it is reasonable to expect that, outside the disk or radius $\delta$ centered at $0$, the measures $P^a(\cdot \, \vert \, 0 \longleftrightarrow \partial B_{\varepsilon}(0))$ and $P^a(\cdot \, \vert \, \mathcal{A}^a_{\eta,\varepsilon}(0))$ are close to each other, in some sense. The lemma below makes this statement precise in a way that is useful in the scaling limit, $a \to 0$, when the statement above is less intuitively clear, since the event $0 \longleftrightarrow \partial B_{\varepsilon}(0)$ becomes an event of measure zero.

A key observation is that, for any $\delta>\eta$, an open circuit $\gamma^a$ in the annulus $A_{\eta,\delta}(0) = B_{\eta}(0) \setminus B_{\delta}(0)$ acts as a ``stopping set'' (the spatial analog of a stopping time) in the sense that, if one considers events that depend only on what is outside $B_{\varepsilon}(0)$, conditioning on $\gamma^a$ being open, on the configuration inside $\gamma^a$ and on either $0 \longleftrightarrow \partial B_{\varepsilon}(0)$ or $\mathcal{A}^a_{\eta,\varepsilon}(0)$ is equivalent to conditioning on the event $\gamma^a \longleftrightarrow \partial B_{\varepsilon}(0)$ (see \eqref{eq:observation} and \eqref{eq:analog-observation}).

With this is mind, the idea of the proof of the lemma is to construct two configurations, $\tilde\Lambda^a$ and $\hat\Lambda^a$, distributed according to $P^a(\cdot \, \vert \, 0 \longleftrightarrow \partial B_{\varepsilon}(0))$ and $P^a(\cdot \, \vert \, \mathcal{A}^a_{\eta,\varepsilon}(0))$, respectively, using the fact that each of those two measures dominates the unconditional measure $P^a$. If $\tilde\Lambda^a$ and $\hat\Lambda^a$ are constructed starting from the \emph{same} percolation configuration ($\Lambda^a$, distributed according to $P^a$), then they are coupled in such a way that the presence of an open circuit $\gamma^a$ in $\Lambda^a$ implies that the same circuit is open in both $\tilde\Lambda^a$ and $\hat\Lambda^a$.
In order to use \eqref{eq:observation} and  \eqref{eq:analog-observation}, it is important to avoid obtaining information from outside the circuit $\gamma^a$, so to generate the configurations $\tilde\Lambda^a$ and $\hat\Lambda^a$ we fix $\delta$ between $\eta$ and $\varepsilon$ and use an exploration process from $\partial B_{\eta}(0)$ outwards that stops when the innermost open circuit inside $A_{\eta,\delta}(0)$ is found in $\Lambda^a$. (It is standard that such an exploration process exists.) If an open circuit is found, the coupling is ``successful'' and can be easily completed, using \eqref{eq:observation} and \eqref{eq:analog-observation}, to produce two configurations that have the same distribution outside $B_{\delta}(0)$ (as expressed by \eqref{eq:coupling-comparison}).

The lemma is useful in the scaling limit because, due to standard RSW arguments (see \cite{grimmett-book}), $\lim_{\eta \to 0} \liminf_{a \to 0}$ of the probability to find an open circuit inside $A_{\eta,\delta}(0)$ in $\Lambda^a$ (i.e., the probability that the coupling is successful) is 1. In the proof of Theorem \ref{thm:scal-lim-connection-probabilities}, this effectively allows us to replace the conditioning on $0 \longleftrightarrow \partial B_{\varepsilon}(0)$ with one on $\mathcal{A}^a_{\eta,\varepsilon}(0)$, which is well-behaved as $a \to 0$.

\begin{figure}[!ht]
	\begin{center}
		\includegraphics[width=10cm]{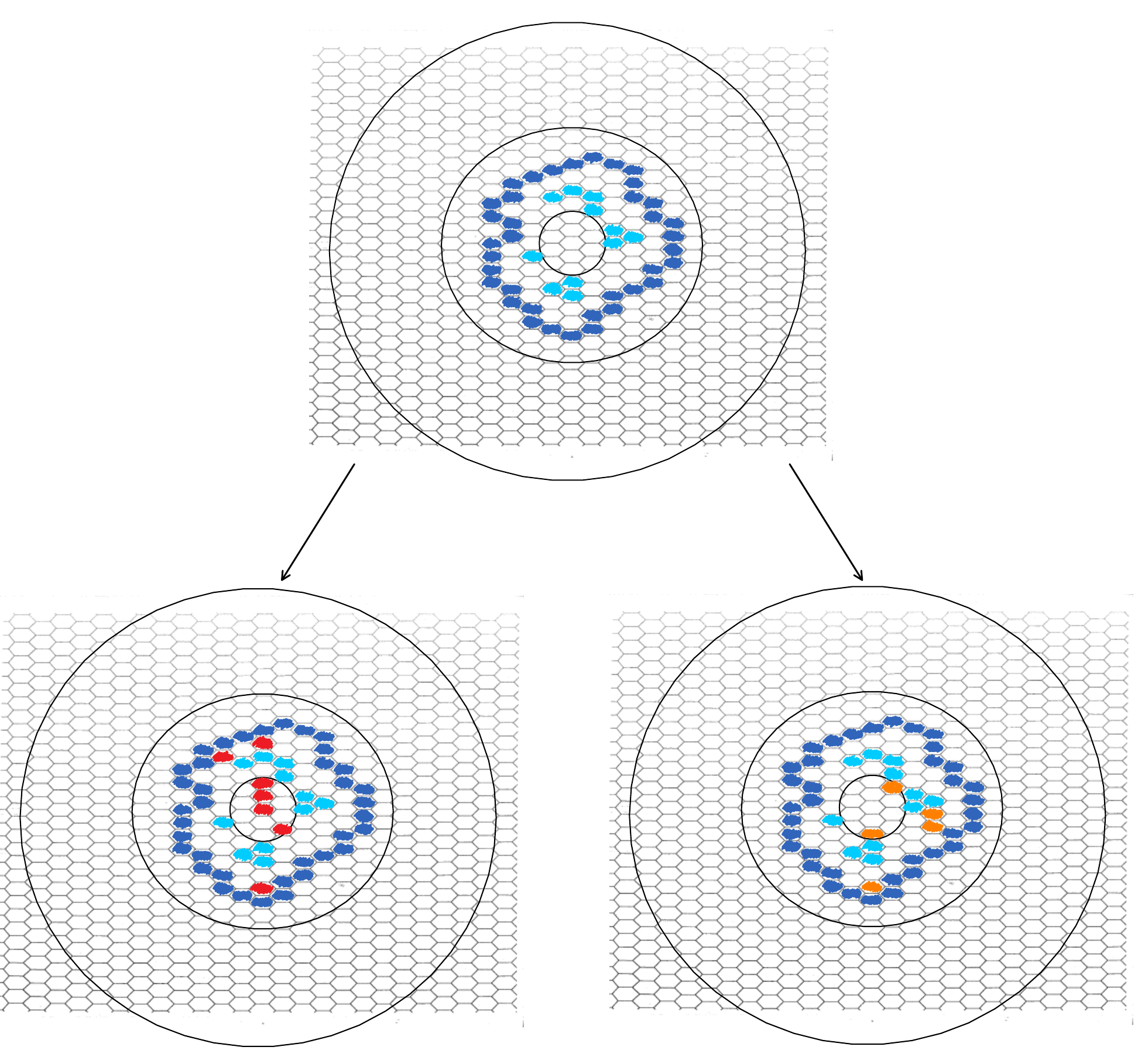}
		\caption{Sketch of one of the steps of the proof of Lemma \ref{lemma:coupling}. In each of the three figures, the three disks represent $B_{\eta}(0) \subset B_{\delta}(0) \subset B_{\varepsilon}(0)$. The top configuration can be generated from $P^a$ by an exploration process inside the annulus $A_{\eta,\delta}(0)$, stopped as soon as the innermost open circuit $\gamma^a$ is found around $B_{\eta}(0)$ (if such a circuit exists). Open hexagons are marked in (light) blue, the hexagons on $\gamma^a$ with a darker shade. The two bottom configurations are obtained from the top one in such a way that hexagons which are open in the top configuration are also open in the bottom ones. Additional open hexagons are marked in red and orange. The left configuration is distributed according to $P^a(\cdot \, \vert \, 0 \longleftrightarrow \partial B_{\varepsilon}(0))$ and the right configuration according to $P^a(\cdot \, \vert \, \mathcal{A}^a_{\eta,\varepsilon}(0))$. The states of hexagons outside $\gamma^a$ have yet to be determined.}
		\label{Fig-Lemma}
	\end{center}
\end{figure}

\begin{lemma} \label{lemma:coupling}
	Consider $\varepsilon>\delta>a$. For any $\delta>\eta>a$, there exists a coupling, ${P}^a_{\eta}$, between $P^a(\cdot \, \vert \, 0 \longleftrightarrow \partial B_{\varepsilon}(0))$ and $P^a(\cdot \, \vert \, \mathcal{A}^a_{\eta,\varepsilon}(0))$, that is, a joint distribution on pairs $(\tilde\Lambda^a,\hat\Lambda^a)$ such that $\tilde\Lambda^a$ and $\hat\Lambda^a$ are distributed according to $P^a(\cdot \, \vert \, 0 \longleftrightarrow \partial B_{\varepsilon}(0))$ and $P^a(\cdot \, \vert \, \mathcal{A}^a_{\eta,\varepsilon}(0))$, respectively, and an event $\mathcal{S}$, such that the following holds: 
	\begin{align}
	P^a_{\eta}(\mathcal{S}) \geq P^a(\exists \text{ open circuit in } A_{\eta,\delta}(0) \text{ surrounding } 0)
	\end{align}
	and, for any event $\mathcal{A}$ that depends only on the states of hexagons of a single percolation configuration outside $B_{\delta}(0)$,
	\begin{align}
	\begin{split} \label{eq:coupling-comparison}
	P^a_{\eta}(\tilde\Lambda^a \in \mathcal{A} \, \vert \, \mathcal{S}) = P^a_{\eta}(\hat\Lambda^a \in \mathcal{A} \, \vert \, \mathcal{S}).
	\end{split}
	\end{align}
\end{lemma}

\begin{proof}
	We start by generating a critical percolation configuration $\Lambda^a$ and letting $\omega(x^a_i):=\mathbf{1}_{\{x^a_i \text{ is open in } \Lambda^a\}}$. Using $\Lambda^a$, we will construct recursively two new percolation configurations, $\tilde\Lambda^a$ and $\hat\Lambda^a$, that are distributed according to $P^a(\cdot \, \vert \, 0 \longleftrightarrow \partial B_{\varepsilon}(0))$ and $P^a(\cdot \, \vert \, \mathcal{A}^a_{\eta,\varepsilon}(0))$, respectively. The construction will start inside the annulus $A_{\eta,\delta}(0)$ and initially we will generate the states of hexagons in $\tilde\Lambda^a$ and $\hat\Lambda^a$ one hexagon at a time.
	
	It is a standard fact \cite{Harris60, Kesten80, Kesten86} that, given a percolation configuration $\Lambda^a$, among all open (simple) circuits in $A_{\eta,\delta}(0)$ surrounding $0$ there is a unique innermost one, i.e., a circuit with minimal interior. Moreover, the event that a circuit $\gamma^a$ is the innermost open circuit depends only on the hexagons in $\gamma^a$ and in its interior, not on the hexagons outside $\gamma^a$. As a consequence, if $\Lambda^a$ contains an open circuit in $A_{\eta,\delta}(0)$ surrounding $0$, one can perform an exploration of the percolation configuration in $A_{\eta,\delta}(0)$ that finds the innermost open circuit without checking the state of any of the hexagons outside it. One can, for example, proceed as follows.
	
	Consider the innermost circuit $G_0$ of hexagons in $A_{\eta,\delta}(0)$ surrounding $0$ (the innermost ``layer'' of hexagons, regardless of their state). The exploration process we describe below will start from $G_0$ and will have the following properties: it stops if an open circuit around $0$ is found or if it reaches $\partial B_{\delta}(0)$; if a closed circuit around $0$ is found, the exploration restarts with the same rules from the innermost circuit $G_k$ of unexplored hexagons outside the closed circuit.
	
	To start the exploration, pick a hexagon from $G_0$ uniformly at random. If the hexagon is open, move counterclockwise to the next hexagon in the circuit, if it is closed, move clockwise to the next hexagon. Repeat this until all the hexagons in $G_0$ are explored or until an interface is found. If an interface is found, explore it using the standard percolation exploration process for interfaces \cite{Sch00, CN07}. If the interface crosses the annulus $A_{\eta,\delta}(0)$, there can be no open circuit surrounding $0$ and the exploration stops. Otherwise, the interface produces an ``excursion'' off $G_0$ into the annulus. Explore the interface until the exploration produces a circuit surrounding $0$ or until it returns to one of the hexagons in $G_0$. In the first case, the circuit can be either open or closed. If the circuit is open, it is the innermost open circuit and the exploration stops. If the circuit is closed, restart the exploration from the innermost (simple) circuit $G_1$ of unexplored hexagons outside the closed circuit. If the exploration of the interface leads back to a hexagon in $G_0$ without generating a circuit, proceed with the rules described above until the exploration either stops or finds a closed circuit around $0$. In the latter case, the exploration process restarts from the innermost circuit $G_1$ of unexplored hexagons outside the closed circuit, with the rules described above.
	The exploration process continues iteratively until either an open circuit around $0$ or an interface crossing the annulus is found, or until it reaches the boundary of $B_{\delta}(0)$.
	
	The algorithm described above produces an ordered sequence $(x^a_i)_{i=1}^N$ of hexagons. Below, we will use that order to generate recursively two additional percolation configurations, as follows (see Figure \ref{Fig-Lemma} for an illustration).
	
	If $(x^a_i)_{i=1}^k$ are the explored hexagons whose states (for the two new configurations) have been generated up to step $k$ of the construction, we let $\tilde\omega(x^a_i):=\mathbf{1}_{\{x^a_i \text{ is open in } \tilde\Lambda^a\}}$ and $\hat\omega(x^a_i):=\mathbf{1}_{\{x^a_i \text{ is open in } \hat\Lambda^a\}}$ for $i=1\ldots,k$. We want to generate the states for the next hexagon in the exploration process, $x^a_{k+1}$, according to the distributions $P^a(\cdot \, \vert \, 0 \longleftrightarrow \partial B_{\varepsilon}(0), (\tilde\omega(x^a_i))_{i=1}^k)$ and $P^a(\cdot \, \vert \, \mathcal{A}^a_{\eta,\varepsilon}(0), (\hat\omega(x^a_i))_{i=1}^k)$, respectively, that is, conditioning on the states generated up to step $k$, so that the new states have the correct distributions.
	
	Note that $P^a(x^a_{k+1} \text{ is open} \, \vert \, 0 \longleftrightarrow \partial B_{\varepsilon}(0), (\tilde\omega(x^a_i))_{i=1}^k)\geq1/2$. To see this, observe that, to evaluate $P^a(x^a_{k+1} \text{ is open} \, \vert \, 0 \longleftrightarrow \partial B_{\varepsilon}(0), (\tilde\omega(x^a_i))_{i=1}^k)$,
	one can consider a graph generated from $a\mathcal{T}$ by removing the vertices $x^a_i$ with $\tilde\omega(x^a_i)=0$ and the edges incident on them. On this new graph, the conditioning is on an increasing event, so that $\mathbf{1}_{\{x^a_{k+1} \text{ is open}\}}$ stochastically dominates a Bernoulli random variable with parameter $1/2$ by the FKG inequality (see \cite{grimmett-book}).
	Therefore, by Strassen's theorem on stochastic domination (see, e.g., \cite{Lindvall99}), there exists a joint distribution, $\tilde{Q}_{k+1}$, on pairs of random variables $(\tilde{X}_{k+1},\tilde\beta_{k+1})$ distributed according to $P^a(\cdot \, \vert \, 0 \longleftrightarrow \partial B_{\varepsilon}(0), (\tilde\omega(x^a_i))_{i=1}^k)$ and a Bernoulli distribution with parameter $1/2$, respectively, such that, if $\tilde\beta_{k+1}=1$ then $\tilde{X}_{k+1}=1$.
	We take $\tilde\omega(x^a_{k+1})$ to be $\tilde{X}_{k+1}$ conditioned on $\tilde\beta_{k+1}=\omega(x^a_{k+1})$, so that $\tilde\omega(x^a_{k+1})$ is distributed according to $\tilde{Q}_{k+1}(\cdot \, \vert \, \tilde\beta_{k+1}=\omega(x^a_{k+1}))$, and declare $x^a_{k+1}$ open in $\tilde\Lambda^a$ if and only if $\tilde\omega(x^a_{k+1})=1$.
	
	Analogously, $P^a(x^a_{k+1} \text{ is open} \, \vert \, \mathcal{A}^a_{\eta,\varepsilon}(0), (\hat\omega(x^a_i))_{i=1}^k)\geq1/2$ and there is a corresponding distribution $\hat{Q}_{k+1}$. We generate $\hat\omega(x^a_{k+1})$ according to $\hat{Q}_{k+1}(\cdot \, \vert \, \hat\beta_{k+1}=\omega(x^a_{k+1}))$ and declare $x^a_{k+1}$ to be open in $\hat\Lambda^a$ if and only if $\hat\omega(x^a_{k+1})=1$.
	
	We apply the construction described above along the sequence $(x^a_i)_{i=1}^N$ of hexagons generated by the exploration process until either a common open circuit $\gamma^a$ around $0$ is produced in $\tilde\Lambda^a$ and $\hat\Lambda^a$ or all hexagons of the sequence $(x^a_i)_{i=1}^N$ have been used. We call $\mathcal{S}$ the event that a common open circuit around $0$ is produced in $\tilde\Lambda^a$ and $\hat\Lambda^a$. 
	
	Note that, by construction, if $\omega(x^a_i)=1$ then $\tilde\omega(x^a_i)=\hat\omega(x^a_i)=1$, that is, if $x^a_i$ is open in $\Lambda^a$, then it is open in both $\tilde\Lambda^a$ and $\hat\Lambda^a$. This implies that, if the algorithm used to generated the sequence $(x^a_i)_{i=1}^N$ stops because an open circuit surrounding the origin is found in $\Lambda^a$, then $\mathcal{S}$ must occur. Therefore,
	\begin{align}
	P^a_{\eta}(\mathcal{S}) \geq P^a(\exists \text{ open circuit in } A_{\eta,\delta}(0) \text{ surrounding } 0).
	\end{align}
	
	
	
	If $\mathcal{S}$ does not occur, we generate configurations in the rest of $B_{\delta}(0)$ and outside $B_{\delta}(0)$ independently, according to $P^a(\cdot \, \vert \, 0 \longleftrightarrow \partial B_{\varepsilon}(0))$ and $P^a(\cdot \, \vert \, \mathcal{A}^a_{\eta,\varepsilon}(0))$, conditioned on the values $\{\tilde\omega(x^a_i)\}_{i=1}^N$ and $\{\hat\omega(x^a_i)\}_{i=1}^N$ previously generated inside $A_{\eta,\delta}(0)$, respectively.
	
	If $\mathcal{S}$ occurs, to complete $\tilde\Lambda^a$ and $\hat\Lambda^a$ inside $\gamma^a$, we generate independent configurations according to $P^a(\cdot \, \vert \, 0 \longleftrightarrow \partial B_{\varepsilon}(0))$
	and $P^a(\cdot \, \vert \, \mathcal{A}^a_{\eta,\varepsilon}(0))$, conditioned on the values $\{\tilde\omega(x^a_i)\}_{i=1}^{N^*}$ and $\{\hat\omega(x^a_i)\}_{i=1}^{N^*}$ previously generated inside $A_{\eta,\delta}(0)$, respectively, where $N^* \leq N$ is the number of hexagons visited until $\gamma^a$ is produced.
	However, for the region outside $\gamma^a$, we generate a single configuration.
	In order to explain how to complete the construction outside $\gamma^a$ when $\mathcal{S}$ occurs, let $\tilde\omega_{\gamma^a}$ and $\hat\omega_{\gamma^a}$ denote the configurations generated inside $\gamma^a$. Outside $\gamma^a$, we need to generate configurations according to the distributions $P^a(\cdot \, \vert \, \tilde\omega_{\gamma^a}, \{\gamma^a \text{ is open}\}, 0 \longleftrightarrow \partial B_{\varepsilon}(0))$ and $P^a(\cdot \, \vert \, \hat\omega_{\gamma^a}, \{\gamma^a \text{ is open}\}, \mathcal{A}_{\eta,\delta}(0))$, respectively.
	
	Now observe that, by construction, the configurations generated inside $\gamma^a$ contain an open path from $0$ to $\gamma^a$ and from $\partial B_{\eta}(0)$ to $\gamma^a$, respectively. Therefore, for any event $\mathcal{A}$ that depends only on the states of hexagons outside $\gamma^a$,
	\begin{align}
	\begin{split} \label{eq:observation}
	& P^a(\mathcal{A} \, \vert \, \tilde\omega_{\gamma^a}, \{\gamma^a \text{ is open}\}, 0 \longleftrightarrow \partial B_{\varepsilon}(0)) \\
	& \quad = \frac{P^a(\mathcal{A}, 0 \longleftrightarrow \partial B_{\varepsilon}(0) \, \vert \, \tilde\omega_{\gamma^a}, \{\gamma^a \text{ is open}\})}{P^a(0 \longleftrightarrow \partial B_{\varepsilon}(0) \, \vert \, \tilde\omega_{\gamma^a}, \{\gamma^a \text{ is open}\})} \\
	& \quad = \frac{P^a(\mathcal{A}, \gamma^a \longleftrightarrow \partial B_{\varepsilon}(0) \, \vert \, \tilde\omega_{\gamma^a}, \{\gamma^a \text{ is open}\})}{P^a(\gamma^a \longleftrightarrow \partial B_{\varepsilon}(0) \, \vert \, \tilde\omega_{\gamma^a}, \{\gamma^a \text{ is open}\})} \\
	& \quad = P^a(\mathcal{A} \, \vert \, \tilde\omega_{\gamma^a}, \{\gamma^a \text{ is open}\}, \gamma^a \longleftrightarrow \partial B_{\varepsilon}(0)) \\
	& \quad = P^a(\mathcal{A} \, \vert \, \gamma^a \longleftrightarrow \partial B_{\varepsilon}(0)).
	\end{split}
	\end{align}
	Analogously,
	\begin{align}
	\begin{split} \label{eq:analog-observation}
	& P^a(\mathcal{A} \, \vert \, \hat\omega_{\gamma^a}, \{\gamma^a \text{ is open}\}, \mathcal{A}_{\eta,\varepsilon}(0)) \\
	& \quad = P^a(\mathcal{A} \, \vert \, \gamma^a \longleftrightarrow \partial B_{\varepsilon}(0)).
	\end{split}
	\end{align}
	
	Therefore, to complete the construction, outside $\gamma^a$ we can generate a single configuration distributed according to $P^a(\cdot \, \vert \, \gamma^a_i \longleftrightarrow \partial B_{\varepsilon}(0))$ and use it for both $\tilde\Lambda^a$ and $\hat\Lambda^a$. This means that $\tilde\Lambda^a$ and $\hat\Lambda^a$ coincide outside $B_{\delta}(0)$ and concludes the proof.
\end{proof}

\medskip

We are now ready to prove the first group of main results.

\medskip

\begin{proof}[Proof of Theorem \ref{thm:scal-lim-connection-probabilities}.]
	By standard RSW arguments (see, e.g., the proofs of Lemmas 2.1 and 2.2 of \cite{CN09}), there are constants $0<K_1<K_2<\infty$, independent of $a$, such that
	\begin{align}
	K_1 \pi_a^n \leq P^a(x^a_1,\ldots,x^a_n) \leq K_2 \pi_a^n,
	\end{align}
	which shows that $\pi_a^{-n} P^a(x^a_1,\ldots,x^a_n)$ stays bounded away from zero and infinity as $a \to 0$.
	
	We fix $\varepsilon>0$ sufficiently small so that $x_1,\ldots,x_n$ are at distance much larger than $\varepsilon$ from each other, and take sequences $x^a_i \in a\mathcal{T}$ such that $x^a_i \to x_i$ for each $i=1,\ldots,n$, as $a \to 0$.
	
	We define the event $x_j \longleftrightarrow x_k$ in the continuum as a countable intersection of decreasing events: $x_j \longleftrightarrow x_k := \cap_{m=1}^{\infty} B_{\delta_m}(x_j) \longleftrightarrow B_{\delta_m}(x_k)$, where $(\delta_m)_{m \geq 1}$ is a decreasing sequence of positive numbers going to zero as $m \to \infty$. 
	For each $\delta_m$ sufficiently small, we are going to establish the existence of the limit of the conditional probability
	\begin{align}
	\begin{split} \label{eq:conditional-probability}
	P^a\big(B_{\delta_m}(x^a_j) \longleftrightarrow B_{\delta_m}(x^a_k), \forall j,k=1,\ldots,n \, \vert \, x^a_i \longleftrightarrow \partial B_{\varepsilon}(x^a_i), \forall i=1,\ldots,n \big)
	\end{split}
	\end{align}
	as $a \to 0$.
	
	For any $a<\eta<\delta_m$, we let $\mathcal{O}^a_{\eta,\delta_m}(x^a_i)$ denote the event that a critical percolation configuration on $a\mathcal{T}$ contains an open circuit surrounding $x^a_i$ in the annulus $A_{\eta,\delta_m}(x^a_i)$. 
	The proof of Lemma \ref{lemma:coupling} can be easily generalized to produce a coupling, $P^a_{\eta;x^a_1,\ldots,x^a_n}$, between configurations $\tilde\Lambda^a$ and $\hat\Lambda^a$ distributed according to $P^a(\cdot \, \vert \, x^a_i \longleftrightarrow \partial B_{\varepsilon}(x^a_i), \forall i=1,\ldots,n)$ and $P^a(\cdot \, \vert \, \mathcal{A}^a_{\eta,\varepsilon}(x^a_i), \forall i=1,\ldots,n)$, respectively, and an event $\mathcal{S}_a$ such that
	\begin{align}
	\begin{split} \label{eq:large-probability}
	& P^a_{\eta;x^a_1,\ldots,x^a_n}\big(\mathcal{S}_a\big) \geq \prod_{i=1}^{n} P^a\big(\mathcal{O}^a_{\eta,\delta_m}(x^a_i)\big) = \Big(P^a\big(\mathcal{O}^a_{\eta,\delta_m}(0)\big)\Big)^n
	\end{split}
	\end{align}
	and, for any event $\mathcal{A}$ that depends only on the states of the hexagons of a single configuration outside  $\cup_{i=1}^{n} B_{\delta_m}(x^a_i)$,
	\begin{align}
	\begin{split} \label{eq:equality}
	& P^a_{\eta;x^a_1,\ldots,x^a_n}\big(\tilde\Lambda^a \in \mathcal{A} \, \vert \, \mathcal{S}_a\big) = P^a_{\eta;x^a_1,\ldots,x^a_n}\big(\hat\Lambda^a \in \mathcal{A} \, \vert \, \mathcal{S}_a\big).
	\end{split}
	\end{align}
	
	Letting $\mathcal{S}_a^c$ denote the complement of $\mathcal{S}_a$ and using \eqref{eq:equality} and the fact that the event $\{ B_{\delta_m}(x^a_j) \longleftrightarrow B_{\delta_m}(x^a_k), \forall j,k=1,\ldots,n \}$ depends only on hexagons outside the disks $B_{\delta_m}(x^a_i)$, we can write
	\begin{align}
	\begin{split} \label{eq:switching}
	& P^a\big(B_{\delta_m}(x^a_j) \longleftrightarrow B_{\delta_m}(x^a_k), \forall j,k=1,\ldots,n \, \vert \, x^a_i \longleftrightarrow \partial B_{\varepsilon}(x^a_i), \forall i=1,\ldots,n \big) \\
	& = P^a_{\eta;x^a_1,\ldots,x^a_n}(\tilde\Lambda^a \in B_{\delta_m}(x^a_j) \longleftrightarrow B_{\delta_m}(x^a_k), \forall j,k \, \vert \, \mathcal{S}_a) P^a_{\eta;x^a_1,\ldots,x^a_n}\big(\mathcal{S}_a\big) \\
	& \quad + P^a_{\eta;x^a_1,\ldots,x^a_n}(\tilde\Lambda^a \in B_{\delta_m}(x^a_j) \longleftrightarrow B_{\delta_m}(x^a_k), \forall j,k \, \vert \, \mathcal{S}_a^c) P^a_{\eta;x^a_1,\ldots,x^a_n}\big(\mathcal{S}_a^c\big) \\
	& = P^a_{\eta;x^a_1,\ldots,x^a_n}(\hat\Lambda^a \in B_{\delta_m}(x^a_j) \longleftrightarrow B_{\delta_m}(x^a_k), \forall j,k \, \vert \, \mathcal{S}_a) P^a_{\eta;x^a_1,\ldots,x^a_n}\big(\mathcal{S}_a\big) \\
	& \quad + P^a_{\eta;x^a_1,\ldots,x^a_n}(\tilde\Lambda^a \in B_{\delta_m}(x^a_j) \longleftrightarrow B_{\delta_m}(x^a_k), \forall j,k \, \vert \, \mathcal{S}_a^c) P^a_{\eta;x^a_1,\ldots,x^a_n}\big(\mathcal{S}_a^c\big) \\
	& = P^a\big(B_{\delta_m}(x^a_j) \longleftrightarrow B_{\delta_m}(x^a_k), \forall j,k \, \vert \, \mathcal{A}_{\eta,\varepsilon}(x^a_i), \forall i \big) \\
	& \quad + P^a_{\eta;x^a_1,\ldots,x^a_n}\big(\mathcal{S}_a^c\big)
	\Big[ P^a_{\eta;x^a_1,\ldots,x^a_n}(\tilde\Lambda^a \in B_{\delta_m}(x^a_j) \longleftrightarrow B_{\delta_m}(x^a_k), \forall j,k \, \vert \, \mathcal{S}_a^c) \\
	& \qquad - P^a_{\eta;x^a_1,\ldots,x^a_n}(\hat\Lambda^a \in B_{\delta_m}(x^a_j) \longleftrightarrow B_{\delta_m}(x^a_k), \forall j,k \, \vert \, \mathcal{S}_a^c) \Big].
	\end{split}
	\end{align}
	
	For each $m \in \mathbb{N}$ such that $\delta_m<\varepsilon$ and every $0<\eta<\delta_m$, using the convergence of percolation interfaces in the scaling limit \cite{CN06}, \eqref{eq:switching} implies that
	\begin{align}
	\begin{split} \label{eq:limsup}
	& \limsup_{a \to 0} P^a\big(B_{\delta_m}(x^a_j) \longleftrightarrow B_{\delta_m}(x^a_k), \forall j,k=1,\ldots,n \, \vert \, x^a_i \longleftrightarrow \partial B_{\varepsilon}(x^a_i), \forall i \big) \\
	& \quad \leq \mathbb{P}\big(B_{\delta_m}(x_j) \longleftrightarrow B_{\delta_m}(x_k), \forall j,k=1,\ldots,n \, \vert \, \mathcal{A}_{\eta,\varepsilon}(x_i), \forall i=1,\ldots,n \big) \\
	& \qquad + \limsup_{a \to 0} \Big(1 - P^a_{\eta;x^a_1,\ldots,x^a_n}\big(\mathcal{S}_a\big)\Big)
	\end{split}
	\end{align}
	and
	\begin{align}
	\begin{split} \label{eq:liminf}
	& \liminf_{a \to 0} P^a\big(B_{\delta_m}(x^a_j) \longleftrightarrow B_{\delta_m}(x^a_k), \forall j,k=1,\ldots,n \, \vert \, x^a_i \longleftrightarrow \partial B_{\varepsilon}(x^a_i), \forall i \big) \\
	& \quad \geq \mathbb{P}\big(B_{\delta_m}(x_j) \longleftrightarrow B_{\delta_m}(x_k), \forall j,k=1,\ldots,n \, \vert \, \mathcal{A}_{\eta,\varepsilon}(x_i), \forall i=1,\ldots,n \big) \\
	& \qquad - \liminf_{a \to 0} \Big(1 - P^a_{\eta;x^a_1,\ldots,x^a_n}\big(\mathcal{S}_a\big)\Big).
	\end{split}
	\end{align}
	It follows that
	\begin{align}
	\begin{split} \label{eq:liminf-limsup}
	& \liminf_{a \to 0} P^a\big(B_{\delta_m}(x^a_j) \longleftrightarrow B_{\delta_m}(x^a_k), \forall j,k \, \vert \, x^a_i \longleftrightarrow \partial B_{\varepsilon}(x^a_i), \forall i \big) \\
	& \quad \geq \limsup_{a \to 0} P^a\big(B_{\delta_m}(x^a_j) \longleftrightarrow B_{\delta_m}(x^a_k), \forall j,k \, \vert \, x^a_i \longleftrightarrow \partial B_{\varepsilon}(x^a_i), \forall i \big) \\
	& \qquad - \Big(1 - \liminf_{a \to 0} P^a_{\eta;x^a_1,\ldots,x^a_n}\big(\mathcal{S}_a\big)\Big) - \Big(1 - \limsup_{a \to 0} P^a_{\eta;x^a_1,\ldots,x^a_n}\big(\mathcal{S}_a\big)\Big).
	\end{split}
	\end{align}
	
	Standard RSW arguments (see, for example, \cite{grimmett-book}) imply that
	\begin{align}
    \lim_{\eta \to 0} \liminf_{a \to 0} P^a(\mathcal{O}^a_{\eta,\delta_m}(x^a_i)) = 1.
	\end{align}
	Therefore, sending $\eta$ to zero in \eqref{eq:liminf-limsup} and using \eqref{eq:large-probability} shows that the limit as $a \to 0$ of \eqref{eq:conditional-probability} exists. Moreover, \eqref{eq:liminf} and \eqref{eq:limsup} imply that
	\begin{align}
	\begin{split} \label{def:conditional-prob}
	& \lim_{a \to 0} P^a\big(B_{\delta_m}(x^a_j) \longleftrightarrow B_{\delta_m}(x^a_k), \forall j,k \, \vert \, x^a_i \longleftrightarrow \partial B_{\varepsilon}(x^a_i), \forall i \big) \\
	& \quad = \lim_{\eta \to 0} \mathbb{P}\big(B_{\delta_m}(x_j) \longleftrightarrow B_{\delta_m}(x_k), \forall j,k \, \vert \, \mathcal{A}_{\eta,\varepsilon}(x_i), \forall i \big) \\
	& \quad =: \mathbb{P}\big(B_{\delta_m}(x_j) \longleftrightarrow B_{\delta_m}(x_k), \forall j,k \, \vert \, \mathcal{A}_{0,\varepsilon}(x_i), \forall i \big) \\
	& \quad \equiv \mathbb{P}\big(B_{\delta_m}(x_j) \longleftrightarrow B_{\delta_m}(x_k), \forall j,k \, \vert \, x_i \longleftrightarrow \partial B_{\varepsilon}(x_i), \forall i \big).
	\end{split}
	\end{align}
	
	\begin{remark} \label{rem:annuli}
		The choice of disks $B_{\varepsilon}(x^a_i)$ and annuli $A_{\eta,\varepsilon}(x^a_i)$ in the argument above is not essential. Any choice of simply connected sets $\tilde{B}_{\varepsilon}(x^a_i)$ and of ``annuli'' $\tilde{A}_{\eta,\delta_m}(x^a_i) = \tilde{B}_{\delta_m}(x^a_i) \setminus \tilde{B}_{\eta}(x^a_i)$, where $\tilde{B}_{\eta}(x^a_i) \subset \tilde{B}_{\delta_m}(x^a_i)$ are sets centered at $x^a_i$ of diameters $2\eta$ and $2\delta_m$, would work. This follows from the observation that, using standard RSW arguments, one can show that $\lim_{\eta \to 0} \liminf_{a \to 0} P^a(\tilde{\mathcal{O}}^a_{\eta,\delta_m}(x^a_i)) = 1$, where $\tilde{\mathcal{O}}^a_{\eta,\delta_m}(x^a_i)$ is the event that $\tilde{A}_{\eta,\delta_m}(x^a_i)$ contains an open circuit surrounding $x^a_i$. 
	\end{remark}
	
	Using the FKG inequality for increasing events, standard RSW arguments imply that there are constants $c^a_m$ such that
	\begin{align}
	\begin{split} \label{eq:sandwich}
	& P^a\big(B_{\delta_m}(x^a_j) \longleftrightarrow B_{\delta_m}(x^a_k), \forall j,k \, \vert \, x^a_i \longleftrightarrow \partial B_{\varepsilon}(x^a_i), \forall i \big) \\
	& \quad \geq P^a\big(x^a_j \longleftrightarrow x^a_k, \forall j,k \, \vert \, x^a_i \longleftrightarrow \partial B_{\varepsilon}(x^a_i), \forall i \big) \\
	& \qquad \geq c^a_m P^a\big(B_{\delta_m}(x^a_j) \longleftrightarrow B_{\delta_m}(x^a_k), \forall j,k \, \vert \, x^a_i \longleftrightarrow \partial B_{\varepsilon}(x^a_i), \forall i \big),
	\end{split}
	\end{align}
	where $c^a_m$ can be taken to be the probability of an open circuit in $A_{\delta_m,\varepsilon}(0)$ surrounding $0$, so that $c^a_m \leq 1$ and $\lim_{m \to \infty}\liminf_{a \to 0} c^a_m = 1$ (see Figure \ref{fig-RSW+FKG}). Using \eqref{def:conditional-prob}, this implies that
	\begin{align}
	\begin{split} \label{eq:sandwiched-limit}
	& \lim_{a \to 0} P^a\big(x^a_j \longleftrightarrow x^a_k, \forall j,k \, \vert \, x^a_i \longleftrightarrow \partial B_{\varepsilon}(x^a_i), \forall i \big) \\
	& \quad = \lim_{m \to \infty} \mathbb{P}\big(B_{\delta_m}(x_j) \longleftrightarrow B_{\delta_m}(x_k), \forall j,k \, \vert \, x_i \longleftrightarrow \partial B_{\varepsilon}(x_i), \forall i \big) \\
	& \quad = \mathbb{P}\big(x_j \longleftrightarrow x_k, \forall j,k \, \vert \, x_i \longleftrightarrow \partial B_{\varepsilon}(x_i), \forall i \big).
	\end{split}
	\end{align}
	
	\begin{figure}[!ht]
		\begin{center}
			\includegraphics[width=7cm]{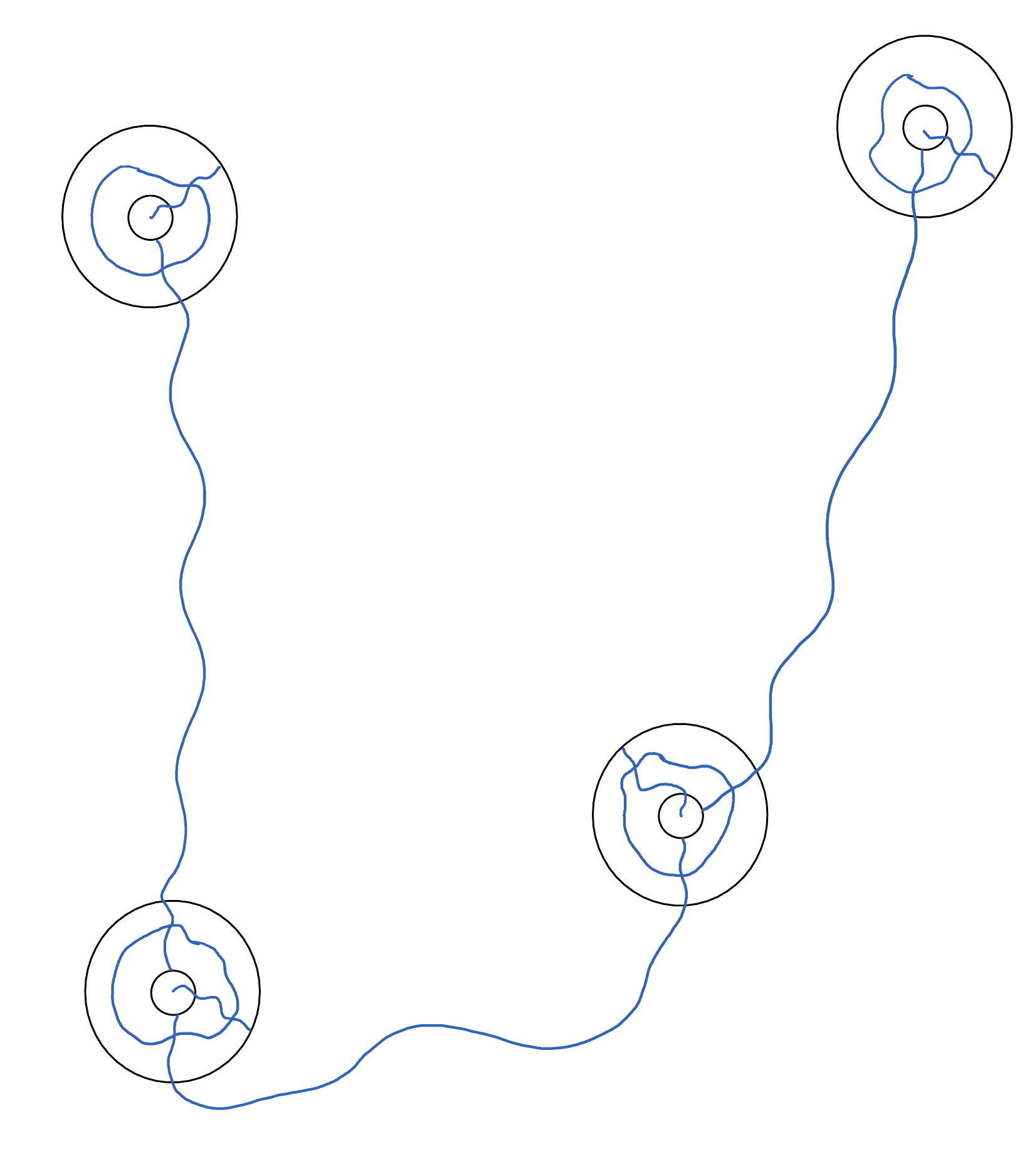}
			\caption{Sketch of the RSW argument used in conjunction with the FKG inequality to obtain \eqref{eq:sandwich}. The four annuli represent the annuli $A_{\delta_m,\varepsilon}(x^a_i)$ and the wiggly lines represent open paths and circuits.}
			\label{fig-RSW+FKG}
		\end{center}
	\end{figure}
	
	The fact that the event $x^a_j \longleftrightarrow x^a_k$ implies the intersection of the independent events $x^a_j \longleftrightarrow \partial B_{\varepsilon}(x^a_j)$ and $x^a_k \longleftrightarrow \partial B_{\varepsilon}(x^a_k)$, combined with \eqref{eq:one-arm-limit}, \eqref{eq:sandwiched-limit} and \eqref{def:conditional-prob}, leads to
	\begin{align}
	\begin{split} \label{eq:lim-pi-P}
	& \lim_{a \to 0} \pi_a^{-n} P^a_{n}(x^a_1,\ldots,x^a_n) \\
	& \quad = \lim_{a \to 0} \pi_a^{-n} P^a\big(x^a_j \longleftrightarrow x^a_k, \forall j,k \, \vert \, x^a_i \longleftrightarrow \partial B_{\varepsilon}(x^a_i), \forall i \big) \\
	& \qquad \prod_{i=1}^{n} P^a\big(x^a_i \longleftrightarrow \partial B_{\varepsilon}(x^a_i)\big) \\
	& \quad = \Big( \prod_{i=1}^{n} \lim_{a \to 0} \pi_a^{-1} P^a\big(x^a_i \longleftrightarrow \partial B_{\varepsilon}(x^a_i)\big) \Big) \\
	& \qquad \lim_{a \to 0} P^a\big(x^a_j \longleftrightarrow x^a_k, \forall j,k \, \vert \, x^a_i \longleftrightarrow \partial B_{\varepsilon}(x^a_i), \forall i \big) \\
	& \quad = \varepsilon^{-5n/48} \, \mathbb{P}\big(x_j \longleftrightarrow x_k, \forall j,k \, \vert \, x_i \longleftrightarrow \partial B_{\varepsilon}(x_i), \forall i \big) \\
	& \quad = \varepsilon^{-5n/48} \, \lim_{m \to \infty} \lim_{\eta \to 0} \mathbb{P}\big(B_{\delta_m}(x_j) \longleftrightarrow B_{\delta_m}(x_k), \forall j,k \, \vert \, \mathcal{A}_{\eta,\varepsilon}(x_i), \forall i \big),
	\end{split}
	\end{align}
	which proves that the limit exists, concluding the first part of the proof.
	Moreover, from the last line of \eqref{eq:lim-pi-P}, we can see that the limit has the same invariance properties as $\mathbb{P}$ (i.e., as the distribution of CLE$_6$) under translations, rotations and reflections, which map disks into disks of the same size.
	
	In order to prove scale covariance, we first note that the discussion above is valid for all $\varepsilon>0$ sufficiently small.
	Now consider a scale transformation $x \mapsto x'=sx$, for some $s>0$, and take $\varepsilon$ so small that \eqref{eq:lim-pi-P} is still valid when $\varepsilon$ is replaced by $\varepsilon/s$. Then, \eqref{eq:lim-pi-P} and the scale invariance of $\mathbb{P}$ imply
	\begin{align}
	\begin{split}
	& P_n(x'_1,\ldots,x'_n) = \varepsilon^{-5n/48} \, \mathbb{P}\big(x'_j \longleftrightarrow x'_k, \forall j,k \, \vert \, x'_i \longleftrightarrow \partial B_{\varepsilon}(x'_i), \forall i \big) \\
	& \quad = s^{-5n/48} (\varepsilon/s)^{-5n/48} \, \mathbb{P}\big(x_j \longleftrightarrow x_k, \forall j,k \, \vert \, x_i \longleftrightarrow \partial B_{\varepsilon/s}(x_i), \forall i \big) \\
	& \quad = s^{-5n/48} P_n(x_1,\ldots,x_n),
	\end{split}
	\end{align}
	as desired.
	
	We now identify $\mathbb{R}^2$ with the complex plane $\mathbb{C}$ and consider a generic M\"obius transformation,
	\begin{align}
	M(x)=\frac{ax+b}{cx+d},
	\end{align}
	with $ad-bc \neq 0$. If $c \neq 0$, $M$ can always be decomposed into the following sequence of transformations (see, for example, \cite{Needham}):
	\begin{enumerate}
		\item $x \mapsto x+\frac{d}{c}$,
		\item $x \mapsto 1/x$,
		\item $x \mapsto \frac{bc-ad}{c^2}x$,
		\item $x \mapsto x+\frac{a}{c}$.
	\end{enumerate}
	In addition, letting $\iota$ denote the imaginary unit and writing $x=\rho e^{\iota\theta}$, the inversion map $x \mapsto 1/x$ can be decomposed into a circle inversion and a reflection in the real axis (i.e., complex conjugation):
	\begin{align} \label{inversion}
	\begin{split}
	& (1) \; x=\rho e^{\iota\theta} \mapsto \text{Inv}(x) := \frac{1}{\rho} e^{\iota\theta} = 1/\bar{x}, \\
	& (2) \; 1/\bar{x} \mapsto \overline{1/\bar{x}}=1/x.
	\end{split}
	\end{align}
	
	Since we have already shown that $P_n$ is invariant under translations, rotations and reflections, and scales covariantly under scale transformations, in order to obtain full covariance under all M\"obius transformations, it suffices to prove covariance under circle inversion. Using the invariance under circle inversion of $\mathbb{P}$ \cite{GMQ21} and \eqref{eq:lim-pi-P}, and assuming that $x_1,\ldots,x_n \neq 0$, we have that, for any $\varepsilon$ sufficiently small,
	\begin{align}
	\begin{split} \label{eq:PInv}
	& P_n(\text{Inv}(x_1),\ldots,\text{Inv}(x_n)) \\
	& \quad = \varepsilon^{-5n/48} \, \mathbb{P}\big(\text{Inv}(x_j) \longleftrightarrow \text{Inv}(x_k), \forall j,k \, \vert \, \text{Inv}(x_i) \longleftrightarrow \partial B_{\varepsilon}(\text{Inv}(x_i)), \forall i \big) \\
	& \quad = \varepsilon^{-5n/48} \, \mathbb{P}\big(x_j \longleftrightarrow x_k, \forall j,k \, \vert \, x_i \longleftrightarrow \text{Inv}\big(\partial B_{\varepsilon}(\text{Inv}(x_i))\big), \forall i \big).
	\end{split}
	\end{align}
	
	Note that the circle inversion, $\text{Inv}$, maps circles to generalized circles (i.e., either a circle or a straight line), but does not necessarily map centers to centers (unless the center of the circle is $0$). This means that, in general, $\text{Inv}\big(\partial B_{\varepsilon}(\text{Inv}(x_i))\big)$ is a circle whose center is not $x_i$. Given a collection of points, $x_1,\ldots,x_n \neq 0$, we can take $\varepsilon$ so small that $0 \notin B_{\varepsilon}(x_i)$ for each $i=1,\ldots,n$. In this case, $\text{Inv}\big(\partial B_{\varepsilon}(\text{Inv}(x_i))\big)$ is a circle whose interior is $\text{Inv}\big(B_{\varepsilon}(\text{Inv}(x_i))\big)$.
	
	We can focus on the case $x_i=\rho_i e^{\iota\theta_i}$ with $\rho_i \geq 1$, sketched in Figure \ref{fig-circ-inv}, since the case $\rho_i \leq 1$ is completely analogous. We take $\varepsilon<1/\rho_i$ for each $i=1,\ldots,n$. Then, the straight line passing through $x_i$ and $\text{Inv}(x_i)$ has two intersections with the circle $\partial B_{\varepsilon}(\text{Inv}(x_i))$:
	\begin{align}
	x_i^1 \equiv (1/\rho_i-\varepsilon) e^{\iota\theta_i} \text{ and } x_i^2 \equiv (1/\rho_i+\varepsilon) e^{\iota\theta_i}.
	\end{align}
	Their circle inversions are
	\begin{align}
	\text{Inv}(x_i^1) = \frac{\rho_i}{1-\rho_i\varepsilon} e^{\iota\theta_i} \text{ and } \text{Inv}(x_i^2) = \frac{\rho_i}{1+\rho_i\varepsilon} e^{\iota\theta_i}
	\end{align}
	and the distances of $\text{Inv}(x_i^1)$ and $\text{Inv}(x_i^2)$ from $x_i$ are
	\begin{align}
	d_i^1 \equiv \rho_i^2 \frac{\varepsilon}{1-\rho_i\varepsilon} \text{ and } d_i^2 \equiv \rho_i^2 \frac{\varepsilon}{1+\rho_i\varepsilon},
	\end{align}
	respectively (see Figure \ref{fig-circ-inv}).
	Therefore, we have that
	\begin{align} \label{eq:inclusions}
	B_{d_i^2}(x_i) \subset \text{Inv}\big(B_{\varepsilon}(\text{Inv}(x_i))\big) \subset B_{d_i^1}(x_i).
	\end{align}
	
	\begin{figure}[!ht]
		\begin{center}
			\includegraphics[width=10cm]{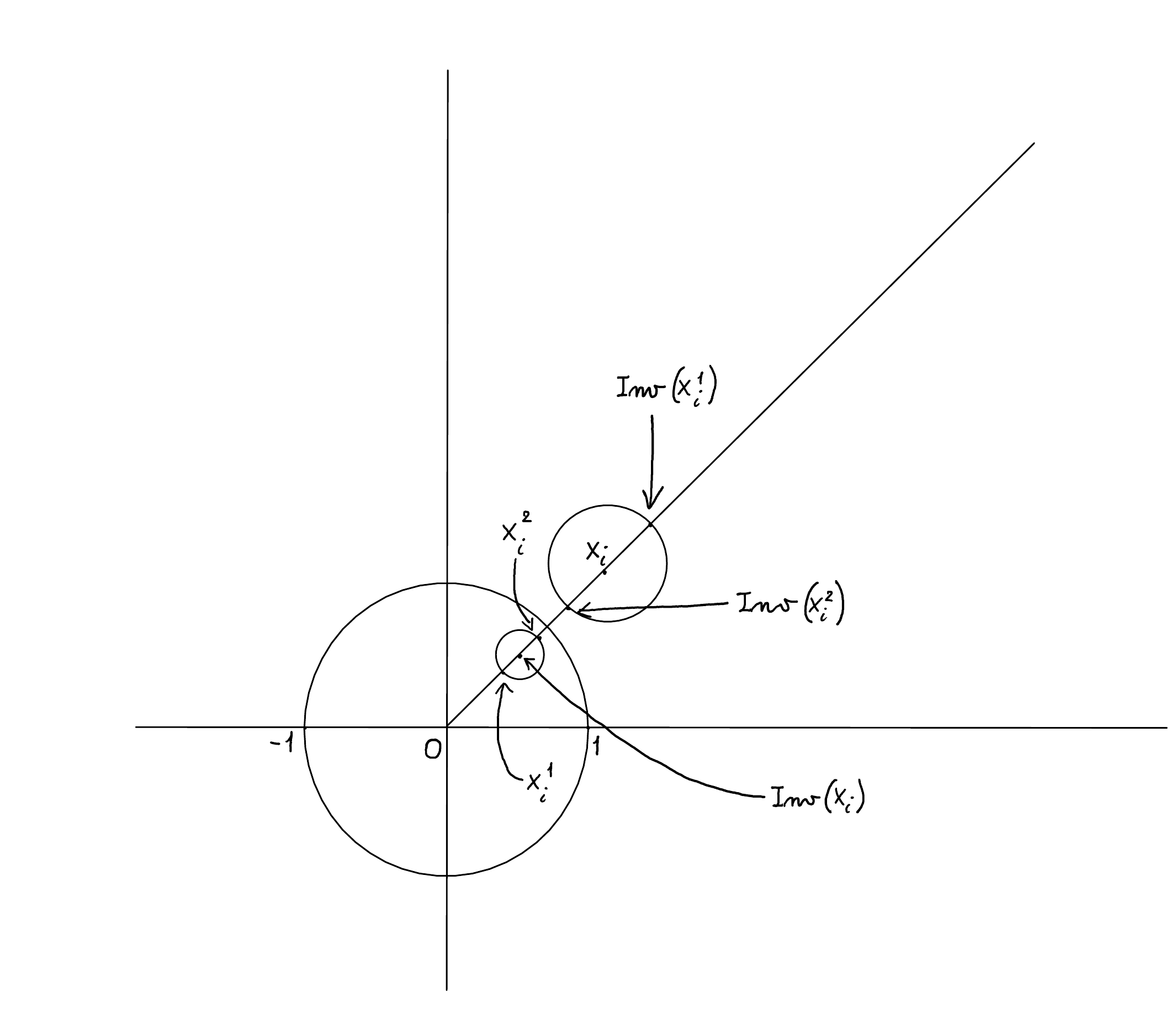}
			\caption{Sketch of the circle inversion of a disk with the notation used in the main text. The small circle around $\text{Inv}(x_i)$ is centered at $\text{Inv}(x_i)$ and has radius $\varepsilon$. The circle centered at 0 has radius 1. The third circle is the inversion of the circle of radius $\varepsilon$ and is not centered at $x_i$.}
			\label{fig-circ-inv}
		\end{center}
	\end{figure}
	
	Now note that, if $x_i \in \Delta^{x_i}_2 \subset \Delta^{x_i}_1$, where $\Delta^{x_i}_1$ and $\Delta^{x_i}_2$ are disks, then
	\begin{align} \label{eq:bound}
	\mathbb{P}\big(x_j \longleftrightarrow x_k, \forall j,k \, \vert \, x_i \longleftrightarrow \partial \Delta^{x_i}_2, \forall i \big) \leq \mathbb{P}\big(x_j \longleftrightarrow x_k, \forall j,k \, \vert \, x^a_i \longleftrightarrow \partial \Delta^{x_i}_1, \forall i \big).
	\end{align}
	To see why this is true, let's first consider two disks, $\Delta_1$ and $\Delta_2$, such that $0 \in \Delta_2 \subset \Delta_1$. For any $x^a \in a\mathcal{T}$ not contained $\Delta_1$, we have that
	\begin{align}
	& P^a\big(x^a \longleftrightarrow 0 \, \vert \, 0 \longleftrightarrow \partial\Delta_2 \big) = \frac{P^a(x^a \longleftrightarrow 0)}{P^a(0 \longleftrightarrow \partial\Delta_2)} \\
	& \qquad \leq \frac{P^a(x^a \longleftrightarrow 0)}{P^a(0 \longleftrightarrow \partial\Delta_1)} = P^a\big(x^a \longleftrightarrow 0 \, \vert \, 0 \longleftrightarrow \partial\Delta_1 \big).
	\end{align}
	The same argument works with $n$ points and pairs of disks, $x_i \in \Delta^{x_i}_2 \subset \Delta^{x_i}_1$ for $i=1,\ldots,n$, assuming that $\Delta^{x_j}_1 \cap \Delta^{x_k}_1=\emptyset$ for $j \neq k$. This leads to the inequality
	\begin{align} \label{eq:conditioning-nested-domains}
	\begin{split}
	& P^a\big(x^a_j \longleftrightarrow x^a_k, \forall j,k \, \vert \, x^a_i \longleftrightarrow \partial \Delta^{x_i}_2, \forall i \big) \\
	& \qquad \leq P^a\big(x^a_j \longleftrightarrow x^a_k, \forall j,k \, \vert \, x^a_i \longleftrightarrow \partial \Delta^{x_i}_1, \forall i \big).
	\end{split}
	\end{align}
	Letting $a \to 0$, and using arguments analogous to those leading to \eqref{eq:sandwiched-limit} to show the existence of the limits (see Remark \ref{rem:annuli}), gives \eqref{eq:bound}.
	
	Combining \eqref{eq:bound} and \eqref{eq:inclusions} with \eqref{eq:PInv} and \eqref{inversion} gives
	\begin{align}
	\begin{split} \label{eq:PInv-prebounds}
	& \varepsilon^{-5n/48} \, \mathbb{P}\big(x_j \longleftrightarrow x_k, \forall j,k \, \vert \, x_i \longleftrightarrow \partial B_{d^2_i}(x_i), \forall i \big) \\
	& \leq P_n(1/x_1,\ldots,1/x_n) \leq \varepsilon^{-5n/48} \, \mathbb{P}\big(x_j \longleftrightarrow x_k, \forall j,k \, \vert \, x_i \longleftrightarrow \partial B_{d^1_i}(x_i), \forall i \big),
	\end{split}
	\end{align}
	which can be written as
	\begin{align}
	\begin{split} \label{eq:PInv-bounds}
	& \Big( \prod_{i=1}^{n} \Big(\frac{1+\rho_i\varepsilon}{\rho_i^{2}} \, d_i^2\Big)^{-5/48} \Big)\mathbb{P}\big(x_j \longleftrightarrow x_k, \forall j,k \, \vert \, x_i \longleftrightarrow \partial B_{d^2_i}(x_i), \forall i \big) \\
	& \leq P_n(1/x_1,\ldots,1/x_n) \\
	& \leq \Big( \prod_{i=1}^{n} \Big(\frac{1-\rho_i\varepsilon}{\rho_i^{2}} \, d_i^1\Big)^{-5/48} \Big) \mathbb{P}\big(x_j \longleftrightarrow x_k, \forall j,k \, \vert \, x_i \longleftrightarrow \partial B_{d^1_i}(x_i), \forall i \big).
	\end{split}
	\end{align}
	
	We now observe that the arguments leading to \eqref{eq:lim-pi-P} do not require that the disks $B_{\varepsilon}(x_i)$ have the same diameter. In other words, we could have chosen a collection of disks of different radii, for example $\{ d_i^{1} \}_{i=1}^n$ or $\{ d_i^{2} \}_{i=1}^n$. This means that \eqref{eq:PInv-bounds} can be rewritten~as
	\begin{align}
	\begin{split}
	& \Big( \prod_{i=1}^{n} (\rho_i^{-2})^{-5/48} (1+\rho_i\varepsilon)^{-5/48} \Big) \, P_n(x_1,\ldots,x_n) \leq P_n(1/x_1,\ldots,1/x_n) \\
	& \qquad \leq \Big( \prod_{i=1}^{n} (\rho_i^{-2})^{-5/48} (1-\rho_i\varepsilon)^{-5/48} \Big) \, P_n(x_1,\ldots,x_n).
	\end{split}
	\end{align}
	Since this is true for all sufficiently small $\varepsilon>0$, we conclude that
	\begin{align}
	\begin{split} \label{eq:PInv-final}
	& P_n(1/x_1,\ldots,1/x_n) = \Big( \prod_{i=1}^{n} (\rho_i^{-2})^{-5/48} \Big) \, P_n(x_1,\ldots,x_n),
	\end{split}
	\end{align}
	which completes the proof of the theorem.
\end{proof}

\medskip

\begin{proof}[Proof of Theorem \ref{thm:scal-lim-bounded-domains}.]
	Letting $\mathbb{P}_D$ denote the distribution of the full scaling limit of percolation in $D$ (i.e., of nested CLE$_6$ in $D$) and following the proof of Theorem \ref{thm:scal-lim-connection-probabilities}, we obtain
	\begin{align}
	\begin{split} \label{eq:lim-pi-P-D}
	& \lim_{a \to 0} \pi_a^{-n} P^a_{D,n}(x^a_1,\ldots,x^a_n) \\
	& \quad = \varepsilon^{-5n/48} \, \mathbb{P}_D\big(x_j \longleftrightarrow x_k, \forall j,k=1,\ldots,n, \, \vert \, x_i \longleftrightarrow \partial B_{\varepsilon}(x_i), \forall i \big)
	\end{split}
	\end{align}
	and, using the conformal invariance properties of $\mathbb{P}_D$ and Remark \ref{rem:annuli},
	\begin{align}
	\begin{split}
	& P_{D',n}(\phi(x_1),\ldots,\phi(x_n)) \\
	& \quad = \varepsilon^{-5n/48} \, \mathbb{P}_{D'}\big(\phi(x_j) \longleftrightarrow \phi(x_k), \forall j,k \, \vert \, \phi(x_i) \longleftrightarrow \partial B_{\varepsilon}(\phi(x_i)), \forall i \big) \\
	& \quad = \varepsilon^{-5n/48} \, \mathbb{P}_D\big(x_j \longleftrightarrow x_k, \forall j,k \, \vert \, x_i \longleftrightarrow \phi^{-1}\big(\partial B_{\varepsilon}(\phi(x_i))\big), \forall i \big).
	\end{split}
	\end{align}
	
	Now let $s_i \equiv \phi'(x_i)$ for each $i=1,\ldots,n$ and let $A_{r_i,R_i}(x_i) = B_{R_i}(x_i) \setminus B_{r_i}(x_i)$ denote the thinnest annulus centered at $x_i$ containing the symmetric difference of $\phi^{-1}(B_{\varepsilon}(\phi(x_i)))$ and $B_{\varepsilon/s_i}(x_i)$. Since $\phi^{-1}$ is analytic and $(\phi^{-1})'(\phi(x_i))=1/s_i$, for every $w \in \partial B_{\varepsilon}(\phi(x_i))$, $|x_i-\phi^{-1}(w)| = \varepsilon/s_i + O(\varepsilon^2)$, which implies that
	\begin{align} \label{eq:limits}
	\lim_{\varepsilon \to 0} \frac{r_i}{\varepsilon} = \lim_{\varepsilon \to 0} \frac{R_i}{\varepsilon} = \frac{1}{s_i}.
	\end{align}
	
	Using \eqref{eq:bound}, we have that
	\begin{align}
	\begin{split} \label{eq:PInv-prebounds-D}
	& \Big(\prod_{i=1}^n \Big(\frac{\varepsilon}{r_i}\Big)^{-5/48} r_i^{-5/48}\Big) \, \mathbb{P}_D\big(x_j \longleftrightarrow x_k, \forall j,k \, \vert \, x_i \longleftrightarrow \partial B_{r_i}(x_i), \forall i \big) \\
	& \quad \leq P_{D',n}(\phi(x_1),\ldots,\phi(x_n)) \\
	& \quad \leq \Big(\prod_{i=1}^n \Big(\frac{\varepsilon}{R_i}\Big)^{-5/48} R_i^{-5/48}\Big) \, \mathbb{P}_D\big(x_j \longleftrightarrow x_k, \forall j,k \, \vert \, x_i \longleftrightarrow \partial B_{R_i}(x_i), \forall i \big).
	\end{split}
	\end{align}
	
	Sending $\varepsilon \to 0$ and using \eqref{eq:limits} and \eqref{eq:lim-pi-P} gives
	\begin{align}
	P_{D',n}(\phi(x_1),\ldots,\phi(x_n)) = \Big(\prod_{i=1}^n s_i^{-5/48}\Big) \, P_{D,n}\big(x_1, \ldots, x_n \big),
	\end{align}
	which concludes the proof.
\end{proof}

The proof of Theorem \ref{thm:scal-lim-correlation-functions} is similar to those of Theorems \ref{thm:scal-lim-connection-probabilities} and \ref{thm:scal-lim-bounded-domains}, but it involves non-increasing events and therefore requires some modifications. We sketch it below, highlighting the differences.

\begin{proof}[Proof of Theorem \ref{thm:scal-lim-correlation-functions}]
	We focus on existence of the limit on the plane and on invariance under M\"obius transformations since the same arguments apply to the case of other simply-connected domains and of local conformal maps. Thanks to \eqref{eq:n-point-function}, it is enough to show that, for any partition $(Q_1,\ldots,Q_k) \in \mathcal{Q}(x_1,\ldots,x_n)$, the desired properties are satisfied by the limit
	\begin{equation}
	\lim_{a \to 0} \pi_a^{-n} P^a(G(Q_1,\ldots,Q_k)).
	\end{equation}
	
	Standard RSW arguments show that $\pi_a^{-n} P^a(G(Q_1,\ldots,Q_k))$ is bounded away from zero and infinity as $a \to 0$. Following the proof of Theorem \ref{thm:scal-lim-connection-probabilities}, we fix $\varepsilon>0$ sufficiently small so that $x_1,\ldots,x_n$ are at distance much larger than $\varepsilon$ from each other, and take sequences $x^a_i \in a\mathcal{T}$ such that $x^a_i \to x_i$ for each $i=1,\ldots,n$ as $a \to 0$. Using the fact that the event $x^a_j \longleftrightarrow x^a_k$ implies the intersection of the independent events $x^a_j \longleftrightarrow \partial B_{\varepsilon}(x^a_j)$ and $x^a_k \longleftrightarrow \partial B_{\varepsilon}(x^a_k)$, combined with \eqref{eq:one-arm-limit}, we can write
	\begin{align}
	\begin{split} \label{eq:lim-pi-PG}
	& \lim_{a \to 0} \pi_a^{-n} P^a(G(Q^a_1,\ldots,Q^a_k)) \\
	& \quad = \lim_{a \to 0} \pi_a^{-n} P^a\big(G(Q^a_1,\ldots,Q^a_k) \vert \, x^a_j \longleftrightarrow \partial B_{\varepsilon}(x^a_j), \forall i=1,\ldots,n \big) \\
	& \qquad \prod_{i=1}^{n} P^a\big(x^a_i \longleftrightarrow \partial B_{\varepsilon}(x^a_i)\big) \\
	& \quad = \varepsilon^{-5n/48} \lim_{a \to 0} P^a\big(G(Q^a_1,\ldots,Q^a_k) \, \vert \, x^a_i \longleftrightarrow \partial B_{\varepsilon}(x^a_i), \forall i=i,\ldots,n \big).
	\end{split}
	\end{align}
	
	The proof that $\lim_{a \to 0} P^a\big(G(Q^a_1,\ldots,Q^a_k) \, \vert \, x^a_i \longleftrightarrow \partial B_{\varepsilon}(x^a_i), \forall i=i,\ldots,n \big)$ exists is similar to that of the existence of the first limit in \eqref{eq:sandwiched-limit}, but the event $G(Q^a_1,\ldots,Q^a_k)$ is not an increasing event, so we cannot use the FKG inequality. 
	We focus for simplicity on the case $n=4$ and $(Q^a_1,Q^a_2)=(\{x^a_1,x^a_2\},\{x^a_3,x^a_4\})$, and observe that the event $x^a_1 \longleftrightarrow x^a_2 \centernot\longleftrightarrow x^a_3 \longleftrightarrow x^a_4$ (i.e., $x^a_1$ and $x^a_2$ are in the same open cluster, $x^a_3$ and $x^a_4$ are in the same open cluster, $x^a_2$ and $x^a_3$ are not in the same cluster) implies that, for all $m \in \mathbb{N}$, if we declare closed all the hexagons inside $B_{\delta_m}(x^a_i)$, $i=1,\ldots,4$, there are two disjoint open clusters connecting $B_{\delta_m}(x^a_1)$ to $B_{\delta_m}(x^a_2)$ and $B_{\delta_m}(x^a_3)$ to $B_{\delta_m}(x^a_4)$, respectively. If we denote by $B_{\delta_m}(x^a_1) \longleftrightarrow B_{\delta_m}(x^a_2) \circ B_{\delta_m}(x^a_3) \longleftrightarrow B_{\delta_m}(x^a_4)$ the latter event,
	we can replace \eqref{eq:sandwich} by
	\begin{align}
	\begin{split} \label{eq:new-sandwich}
	& P^a\big(B_{\delta_m}(x^a_1) \longleftrightarrow B_{\delta_m}(x^a_2) \circ B_{\delta_m}(x^a_3) \longleftrightarrow B_{\delta_m}(x^a_4) \, \vert \, x^a_i \longleftrightarrow \partial B_{\varepsilon}(x^a_i), \forall i \big) \\
	& \geq P^a\big(x^a_1 \longleftrightarrow x^a_2 \centernot\longleftrightarrow x^a_3 \longleftrightarrow x^a_4 \, \vert \, x^a_i \longleftrightarrow \partial B_{\varepsilon}(x^a_i), \forall i \big) \\
	& \geq P^a\big(B_{\delta_m}(x^a_1) \longleftrightarrow B_{\delta_m}(x^a_2), B_{\delta_m}(x^a_3) \longleftrightarrow B_{\delta_m}(x^a_4), B_{\delta_k}(x^a_1) \centernot\longleftrightarrow B_{\delta_k}(x^a_3), \\ 
	& \qquad \quad \mathcal{O}_{\delta_m,\varepsilon}(x^a_i), \forall i \, \vert \, x^a_i \longleftrightarrow \partial B_{\varepsilon}(x^a_i), \forall i \big),
	\end{split}
	\end{align}
	where $x^a_1 \centernot\longleftrightarrow x^a_3$ is the event that $x^a_1$ and $x^a_3$ are not in the same open cluster, $\mathcal{O}^a_{\delta_m,\varepsilon}(x_i)$ denotes the event that there is an open circuit in $A_{\delta_m,\varepsilon}(x^a_i)$ surrounding $x^a_i$, and $k \leq m$ (so that $\delta_k\geq\delta_m$).
	
	The proof that
	\begin{align}
	\begin{split} \label{eq:new-conditional-probability}
	& P^a\big(B_{\delta_m}(x^a_1) \longleftrightarrow B_{\delta_m}(x^a_2), B_{\delta_m}(x^a_3) \longleftrightarrow B_{\delta_m}(x^a_4), B_{\delta_k}(x^a_1) \centernot\longleftrightarrow B_{\delta_k}(x^a_3), \\
	& \qquad \mathcal{O}_{\delta_m,\varepsilon}(x^a_i), \forall i \, \vert \, x^a_i \longleftrightarrow \partial B_{\varepsilon}(x^a_i), \forall i \big)
	\end{split}
	\end{align}
	has a limit as $a \to 0$ is the same as the proof, given earlier, that \eqref{eq:conditional-probability} has a limit as $a \to 0$. The fact that the same proof applies to \eqref{eq:new-conditional-probability} follows from the observation that the conditioning is the same and that the event
	\begin{align}
	\begin{split}
	& \{B_{\delta_m}(x^a_1) \longleftrightarrow B_{\delta_m}(x^a_2) \cap B_{\delta_m}(x^a_3) \longleftrightarrow B_{\delta_m}(x^a_4) \cap B_{\delta_k}(x^a_1) \centernot\longleftrightarrow B_{\delta_k}(x^a_3) \\
	& \quad \cap \mathcal{O}_{\delta_m,\varepsilon}(x^a_i), \forall i\}
	\end{split}
	\end{align}
	in \eqref{eq:new-conditional-probability} depends only on hexagons outside the disks $B_{\delta_m}(x^a_i)$, just like the event $\{B_{\delta_m}(x^a_j) \longleftrightarrow B_{\delta_m}(x^a_k), \forall j,k=1,\ldots,n\}$ in \eqref{eq:conditional-probability}.
	
	Moreover, using the fact that
	\begin{align}
	\begin{split}
	& \lim_{m \to \infty}\liminf_{a \to 0}P^a\big(\mathcal{O}^a_{\delta_m,\varepsilon}(x^a_i) \, \vert \, x^a_i \longleftrightarrow \partial B_{\varepsilon}(x^a_i) \big) \\
	& \quad \geq \lim_{m \to \infty} \liminf_{a \to 0} P^a\big(\mathcal{O}^a_{\delta_m,\varepsilon}(x^a_i)\big) = 1,
	\end{split}
	\end{align}
	we see that
	\begin{align}
	\begin{split} \label{eq:limit-conditional-probability1}
	& \lim_{k \to \infty} \lim_{m \to \infty} \lim_{a \to 0} P^a\big(B_{\delta_m}(x^a_1) \longleftrightarrow B_{\delta_m}(x^a_2), B_{\delta_m}(x^a_3) \longleftrightarrow B_{\delta_m}(x^a_4), \\
	&  \qquad \qquad \qquad \quad B_{\delta_k}(x^a_1) \centernot\longleftrightarrow B_{\delta_k}(x^a_3), \mathcal{O}_{\delta_m,\varepsilon}(x^a_i), \forall i \, \vert \, x^a_i \longleftrightarrow \partial B_{\varepsilon}(x^a_i), \forall i \big) \\
	& \quad = \mathbb{P}\big(x_1 \longleftrightarrow x_2, x_3 \longleftrightarrow x_4, x_1 \centernot\longleftrightarrow x_3 \, \vert \, x_i \longleftrightarrow \partial B_{\varepsilon}(x_i), \forall i \big),
	\end{split}
	\end{align}
	where $x_1 \centernot\longleftrightarrow x_3 := \cup_{k=1}^{\infty} B_{\delta_k}(x_1) \centernot\longleftrightarrow B_{\delta_k}(x_3)$ is a countable union of increasing events.
	
	Since the event $B_{\delta_m}(x^a_1) \longleftrightarrow B_{\delta_m}(x^a_2) \circ B_{\delta_m}(x^a_3) \longleftrightarrow B_{\delta_m}(x^a_4)$ depends only on hexagons outside the disks $B_{\delta_m}(x^a_i)$, the proof that
	\begin{align}
	\begin{split}
	P^a\big(B_{\delta_m}(x^a_1) \longleftrightarrow B_{\delta_m}(x^a_2) \circ B_{\delta_m}(x^a_3) \longleftrightarrow B_{\delta_m}(x^a_4) \, \vert \, x^a_i \longleftrightarrow \partial B_{\varepsilon}(x^a_i), \forall i \big)
	\end{split}
	\end{align}
	has a limit as $a \to 0$ is also the same as the proof that \eqref{eq:conditional-probability} has a limit as $a \to 0$. 
	
	Moreover, following the proof that \eqref{eq:conditional-probability} has a limit shows that
	\begin{align}
	\begin{split} \label{eq:limit-conditional-probability2}
	& \lim_{a \to 0} P^a\big(B_{\delta_m}(x^a_1) \longleftrightarrow B_{\delta_m}(x^a_2) \circ B_{\delta_m}(x^a_3) \longleftrightarrow B_{\delta_m}(x^a_4) \, \vert \, x^a_i \longleftrightarrow \partial B_{\varepsilon}(x^a_i), \forall i \big) \\
	& = \lim_{\eta \to 0} \mathbb{P}\big(B_{\delta_m}(x_1) \longleftrightarrow B_{\delta_m}(x_2) \circ B_{\delta_m}(x_3) \longleftrightarrow B_{\delta_m}(x_4) \, \vert \, \mathcal{A}_{\eta,\varepsilon}(x_i), \forall i \big) \\
	& =: \mathbb{P}\big(B_{\delta_m}(x_1) \longleftrightarrow B_{\delta_m}(x_2) \circ B_{\delta_m}(x_3) \longleftrightarrow B_{\delta_m}(x_4) \, \vert \, \mathcal{A}_{0,\varepsilon}(x_i), \forall i \big),
	\end{split}
	\end{align}
	where, after the limit $a \to 0$ is taken, $B_{\delta_m}(x_1) \longleftrightarrow B_{\delta_m}(x_2) \circ B_{\delta_m}(x_3) \longleftrightarrow B_{\delta_m}(x_4)$ should be interpreted as the event that the disks $B_{\delta_m}(x_1)$, $B_{\delta_m}(x_2)$ and the disks $B_{\delta_m}(x_3)$, $B_{\delta_m}(x_4)$ are connected disjointly outside $B_{\delta_m}(x_i), i=1,\ldots,4$.
	
	Now observe that, for any $m$ such that $\delta_m<\ell:=\min_{i\neq j}\frac{\Vert x^a_i-x^a_j \Vert}{2}$ and all $k \geq m$ (so that $\delta_k\leq\delta_m$),
	\begin{align}
	\begin{split} \label{eq:inclusion}
	& \{ B_{\delta_m}(x_1) \longleftrightarrow B_{\delta_m}(x_2) \circ B_{\delta_m}(x_3) \longleftrightarrow B_{\delta_m}(x_4) \} \cap \{ B_{\delta_k}(x_1) \longleftrightarrow B_{\delta_k}(x_3) \} \\
	& \subset \cup_{i=1}^4 \mathcal{F}_{\delta_m,\ell}(x_i),
	\end{split}
	\end{align}
	where $\mathcal{F}_{\delta_m,\ell}(x)$ is a \emph{four-arm event}, corresponding to 
	the presence of four disjoint interfaces crossing the annulus $A_{\delta_m,\ell}(x)$, alternating in direction. This can be understood thinking in terms of (continuum) clusters: if one wants to connect $B_{\delta_m}(x_1)$ to $B_{\delta_m}(x_2)$ and $B_{\delta_m}(x_3)$ to $B_{\delta_m}(x_4)$ disjointly outside $B_{\delta_m}(x_i), i=1,\ldots,4$, while at the same time connecting $B_{\delta_k}(x_1)$ to $B_{\delta_k}(x_3)$, outside the disks $B_{\delta_m}(x_i)$ one needs at least two separate pieces of one or more open clusters getting close to one of the points $x_i$. Those two (or more) pieces must be separated by closed clusters, which implies the presence of four disjoint interfaces crossing $A_{\delta_m,\ell}(x_i)$ (see Fig. \ref{fig-four-arm-event}). Therefore, using \eqref{eq:limit-conditional-probability2}, \eqref{eq:lim-pi-P}, \eqref{eq:inclusion} and the second limit in the third displayed equation on p.~999 of \cite{GPS13}, we can write
	\begin{align}
	\begin{split}
	& \mathbb{P}\big(\{B_{\delta_m}(x_1) \longleftrightarrow B_{\delta_m}(x_2) \circ B_{\delta_m}(x_3) \longleftrightarrow B_{\delta_m}(x_4) \} \\
	& \qquad \cap \{ B_{\delta_k}(x_1) \longleftrightarrow B_{\delta_k}(x_3) \} \, \vert \, x_i \longleftrightarrow \partial B_{\varepsilon}(x_i), \forall i=1,\ldots,4 \big) \\
	& \quad = \lim_{\eta \to 0} \mathbb{P}\big(\{B_{\delta_m}(x_1) \longleftrightarrow B_{\delta_m}(x_2) \circ B_{\delta_m}(x_3) \longleftrightarrow B_{\delta_m}(x_4) \} \\
	& \qquad \qquad \quad \cap \{ B_{\delta_k}(x_1) \longleftrightarrow B_{\delta_k}(x_3) \} \, \vert \, \mathcal{A}_{\eta,\varepsilon}(x_i), \forall i=1,\ldots,4 \big) \\
	& \quad \leq \lim_{\eta \to 0} \frac{\mathbb{P}\big(\{\cup_{i=1}^4 \mathcal{F}_{\delta_m,l}(x_i)\} \cap \{ \mathcal{A}_{\eta,\delta_m}(x_i), \forall i=1,\ldots,4\} \big)}{\Big(\mathbb{P}\big(\mathcal{A}_{\eta,\varepsilon}(0)\big)\Big)^4} \\
	& \quad = 4\mathbb{P}\big(\mathcal{F}_{\delta_m,l}(0)\big) \left(\lim_{\eta \to 0}\frac{\mathbb{P}\big(\mathcal{A}_{\eta,\delta_m}(0)\big)}{\mathbb{P}\big(\mathcal{A}_{\eta,\varepsilon}(0)\big)}\right)^4 = 4\mathbb{P}\big(\mathcal{F}_{\delta_m,l}(0)\big) \left(\frac{\delta_m}{\varepsilon}\right)^{-5/12}.
	\end{split}
	\end{align}
	
	\begin{figure}[!ht]
		\begin{center}
			\includegraphics[width=5cm]{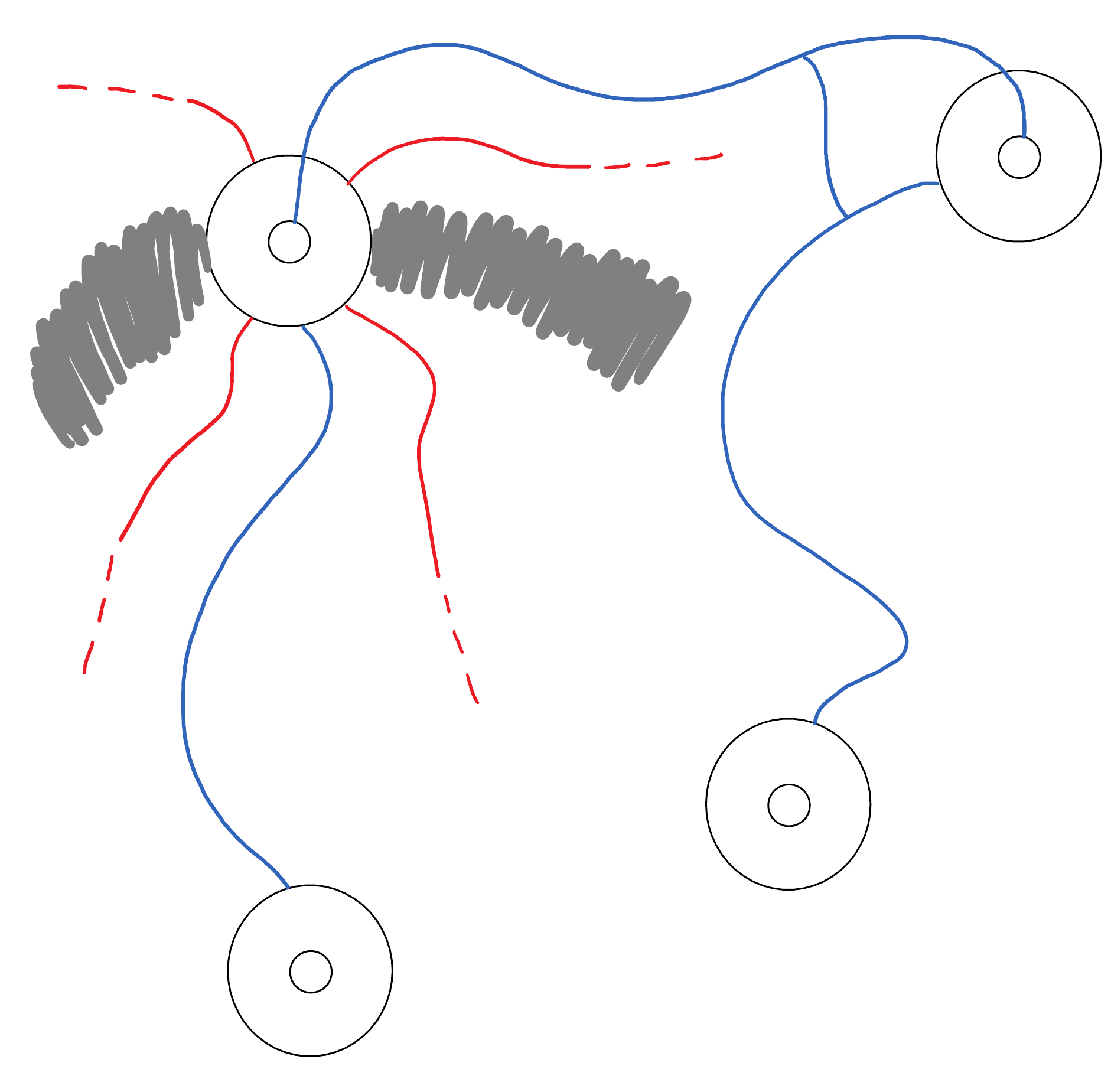}
			\caption{Sketch of the event $\{ B_{\delta_m}(x_1) \longleftrightarrow B_{\delta_m}(x_2) \circ B_{\delta_m}(x_3) \longleftrightarrow B_{\delta_m}(x_4) \} \cap \{ B_{\delta_k}(x_1) \longleftrightarrow B_{\delta_k}(x_3) \}$, with the four-arm event $\mathcal{F}_{\delta_m,l}(x_1)$. The wiggly continuous lines connecting disks represent open paths (possibly belonging to the same open cluster), the shaded regions represent closed clusters, and the partly dashed lines represent interfaces between open and closed clusters.}
			\label{fig-four-arm-event}
		\end{center}
	\end{figure}
	
	
	
	Using the fact that the four-arm event $\mathcal{F}_{\delta_m,\ell}(x_i)$ has probability of order $\big( \frac{\delta_m}{\ell} \big)^{5/4}$ (see, e.g., Remark 4.10 of \cite{GPS13}), the last equation implies that
	\begin{align}
	\begin{split} \label{eq:difference}
	& \lim_{m \to \infty} \mathbb{P}\big(\{B_{\delta_m}(x_1) \longleftrightarrow B_{\delta_m}(x_2) \circ B_{\delta_m}(x_3) \longleftrightarrow B_{\delta_m}(x_4)\} \cap \{ x_1 \longleftrightarrow x_3 \} \, \vert \\
	& \qquad \qquad x_i \longleftrightarrow \partial B_{\varepsilon}(x_i), \forall i=1\ldots,4 \big) \\
	& = \lim_{m \to \infty} \lim_{k \to \infty} \mathbb{P}\big(\{B_{\delta_m}(x_1) \longleftrightarrow B_{\delta_m}(x_2) \circ B_{\delta_m}(x_3) \longleftrightarrow B_{\delta_m}(x_4)\} \\
	& \qquad \qquad \cap \{ B_{\delta_k}(x_1) \longleftrightarrow B_{\delta_k}(x_3) \} \, \vert x_i \longleftrightarrow \partial B_{\varepsilon}(x_i), \forall i=1,\ldots,4 \big) = 0.
	\end{split}
	\end{align}
	Writing $x_1 \longleftrightarrow x_2 \centernot\longleftrightarrow x_3 \longleftrightarrow x_4$ for the countable intersection of events $\cap_{m=1}^{\infty} \{ B_{\delta_m}(x_1) \longleftrightarrow B_{\delta_m}(x_2) \circ B_{\delta_m}(x_1) \longleftrightarrow B_{\delta_m}(x_2) \} \cap \{x_1 \centernot\longleftrightarrow x_3\}$ and observing that the right-hand side is a countable intersection of decreasing events, \eqref{eq:difference} implies that
	\begin{align}
	\begin{split} \label{eq:m-limit}
	& \lim_{m \to \infty} \mathbb{P}\big(B_{\delta_m}(x_1) \longleftrightarrow B_{\delta_m}(x_2) \circ B_{\delta_m}(x_3) \longleftrightarrow B_{\delta_m}(x_4) \, \vert \, \mathcal{A}_{0,\varepsilon}(x_i), \forall i \big) \\
	& \quad = \mathbb{P}\big(x_1 \longleftrightarrow x_2, x_3 \longleftrightarrow x_4, x_1 \centernot\longleftrightarrow x_3 \, \vert \, x_i \longleftrightarrow \partial B_{\varepsilon}(x_i), \forall i \big).
	\end{split}
	\end{align}
	
	Combining \eqref{eq:new-sandwich} with \eqref{eq:limit-conditional-probability1}, 
	\eqref{eq:limit-conditional-probability2} and \eqref{eq:m-limit} gives
	\begin{align}
	\begin{split} \label{eq:limit-exists}
	& \lim_{a \to 0} P^a\big(x^a_1 \longleftrightarrow x^a_2 \centernot\longleftrightarrow x^a_3 \longleftrightarrow x^a_4 \, \vert \, x^a_i \longleftrightarrow \partial B_{\varepsilon}(x^a_i), \forall i=1,\ldots,4 \big) \\
	& \quad = \mathbb{P}\big(x_1 \longleftrightarrow x_2, x_3 \longleftrightarrow x_4, x_1 \centernot\longleftrightarrow x_3 \, \vert \, x_i \longleftrightarrow \partial B_{\varepsilon}(x_i), \forall i=1,\ldots,4 \big) \\
	& \quad = \lim_{m \to \infty} \lim_{\eta \to 0} \mathbb{P}\big(B_{\delta_m}(x_1) \longleftrightarrow B_{\delta_m}(x_2), B_{\delta_m}(x_3) \longleftrightarrow B_{\delta_m}(x_4), x_1 \centernot\longleftrightarrow x_3 \, \vert \\
	& \qquad \qquad \qquad \quad \mathcal{A}_{\eta,\varepsilon}(x_i), \forall i=1,\ldots,4 \big),
	\end{split}
	\end{align}
	which proves that the limit exists. Moreover, \eqref{eq:limit-exists} shows that the limit has the same invariance properties as $\mathbb{P}$ (i.e., as the distribution of CLE$_6$) under translations, rotations and reflections, which map disks into disks of the same size.
	
	In order to prove scale covariance, we first note that the discussion above is valid for any $\varepsilon>0$ sufficiently small.
	Now consider a scale transformation $x \to x'=sx$, for some $s>0$, and take $\varepsilon$ so small that \eqref{eq:lim-pi-PG} is still valid for $\varepsilon/s$.
	Then, using \eqref{eq:lim-pi-PG}, \eqref{eq:limit-exists} and the scale invariance of $\mathbb{P}$, we can write
	\begin{align}
	\begin{split}
	& \varepsilon^{-5/12} \, \mathbb{P}\big(x'_1 \longleftrightarrow x'_2 \centernot\longleftrightarrow x'_3 \longleftrightarrow x'_4 \, \vert \, x'_i \longleftrightarrow \partial B_{\varepsilon}(x'_i), \forall i=i,\ldots,4 \big) \\
	& \quad = s^{-5/12} (\varepsilon/s)^{-5/12} \, \mathbb{P}\big(x_1 \longleftrightarrow x_2 \centernot\longleftrightarrow x_3 \longleftrightarrow x_4 \, \vert \\& \qquad \qquad \qquad \qquad \qquad \quad x_i \longleftrightarrow \partial B_{\varepsilon/s}(x_i), \forall i=i,\ldots,4 \big) \\
	& \quad = s^{-5/12} \lim_{a \to 0} \pi_a^{-4} P^a(x^a_1 \longleftrightarrow x^a_2 \centernot\longleftrightarrow x^a_3 \longleftrightarrow x^a_4),
	\end{split}
	\end{align}
	as desired.
	
	The rest of the proof is also analogous to the proof of Theorem \ref{thm:scal-lim-connection-probabilities}, but with \eqref{eq:bound} replaced by
	\begin{align}
	\begin{split} \label{eq:newbound}
	& \mathbb{P}\big(x_1 \longleftrightarrow 0, x_2 \centernot\longleftrightarrow 0 \, \vert \, 0 \longleftrightarrow \partial \Delta_2, x_1 \longleftrightarrow \partial \Delta^{x_1}_2 \big) \\
	& \qquad \leq \mathbb{P}\big(x_1 \longleftrightarrow 0, x_2 \centernot\longleftrightarrow 0 \, \vert \, 0 \longleftrightarrow \partial \Delta_1, x_1 \longleftrightarrow \partial\Delta^{x_1}_1 \big),
	\end{split}
	\end{align}
	where $0 \in \Delta_2 \subset \Delta_1$, $x_1 \in \Delta^{x_1}_2 \subset \Delta^{x_1}_1$. To see why \eqref{eq:newbound} holds, observe that, if $0,x^a_1,x^a_2$ are sufficiently far from each other, compared to the diameters of $\Delta_1$ and $\Delta^{x_1}_1$, then
	\begin{align}
	\begin{split}
	& P^a\big(x^a_1 \longleftrightarrow 0, x^a_2 \centernot\longleftrightarrow 0 \, \vert \, 0 \longleftrightarrow \partial \Delta_2, x^a_1 \longleftrightarrow \partial \Delta^{x_1}_2 \big) \\
	& \quad = \frac{P^a\big(x^a_1 \longleftrightarrow 0, x^a_2 \centernot\longleftrightarrow 0 \big)}{P^a\big(0 \longleftrightarrow \partial \Delta_2, x^a_1 \longleftrightarrow \partial \Delta^{x_1}_2 \big)} \leq \frac{P^a\big(x^a_1 \longleftrightarrow 0, x^a_2 \centernot\longleftrightarrow 0 \big)}{P^a\big(0 \longleftrightarrow \partial \Delta_1, x^a_1 \longleftrightarrow \partial \Delta^{x_1}_1 \big)} \\
	& \quad = P\big(x^a_1 \longleftrightarrow 0, x^a_2 \centernot\longleftrightarrow 0 \, \vert \, 0 \longleftrightarrow \partial \Delta_1, x^a_1 \longleftrightarrow \partial\Delta^{x_1}_1 \big);
	\end{split}
	\end{align}
	letting $a \to 0$ gives \eqref{eq:newbound}.
	
	Focusing again on the case $n=4$ and $(Q_1,Q_2)=(\{x_1,x_2\},\{x_3,x_4\})$ for simplicity, we can use \eqref{eq:newbound} and \eqref{eq:inclusions} to write
	\begin{align}
	\begin{split} \label{eq:replacement-bounds}
	& \varepsilon^{-5/12} \mathbb{P}\big(x_1 \longleftrightarrow x_2 \centernot\longleftrightarrow x_3 \longleftrightarrow x_4 \, \vert \, x_i \longleftrightarrow \partial B_{d_i^2}(x_i), \forall i \big) \\
	& \leq \varepsilon^{-5/12} \mathbb{P}\big(\text{Inv}(x_1) \longleftrightarrow \text{Inv}(x_2) \centernot\longleftrightarrow \text{Inv}(x_3) \longleftrightarrow \text{Inv}(x_4) \, \vert \\
	& \qquad \qquad \quad \text{Inv}(x_i) \longleftrightarrow \partial B_{\varepsilon}(\text{Inv}(x_i)), \forall i \big) \\
	& \leq \varepsilon^{-5/12} \mathbb{P}\big(x_1 \longleftrightarrow x_2 \centernot\longleftrightarrow x_3 \longleftrightarrow x^a_4 \, \vert x_i \longleftrightarrow \partial B_{d^1_i}(x_i), \forall i \big).
	\end{split}
	\end{align}
	The result for the case $n=4$ and $(Q_1,Q_2)=(\{x_1,x_2\},\{x_3,x_4\})$ now follows from the same arguments as in the proof of Theorem \ref{thm:scal-lim-connection-probabilities}, but with \eqref{eq:PInv-prebounds} is replaced by \eqref{eq:replacement-bounds}. The general case, $n>4$, can be treated in a similar way.
\end{proof}

\subsection{Analysis of the four-point function} \label{sec:4-point-function}

We conclude this section with a brief heuristic discussion of the four-point function 
\begin{align} \label{eq:scal-lim-correlation-functions}
\begin{split}
& \mathrm{C}_{4}(x_1,\ldots,x_4) := \lim_{a \to 0} \pi_a^{-4} \langle S_{x_1^a} \ldots S_{x_{4}^a} \rangle_a \\
& \quad = P_4(x_1,\ldots,x_4) + P_4(x_1,x_2 \vert x_3,x_4) + P_4(x_1,x_3 \vert x_2,x_4) + P_4(x_1,x_4 \vert x_2,x_3),
\end{split}
\end{align}
where we have extended to the scaling limit the notation introduced in Remark \ref{rem:4-point-function}.
Four-point functions are very important in CFT because they contain a wealth of information on the ``operator content'' of the theory \cite{DFMS}, that is, on the primary fields. 

Using \eqref{eq:lim-pi-P}, \eqref{eq:lim-pi-PG} and \eqref{eq:limit-exists}, for any $\varepsilon$ sufficiently small, we have
\begin{align}
\begin{split}
& \mathrm{C}_{4}(x_1,\ldots,x_4) = \varepsilon^{-5/12} \, \big[ \mathbb{P}\big(x_j \longleftrightarrow x_k, \forall j,k=1,\ldots,4, \, \vert \\
& \qquad \qquad \qquad \qquad \qquad \quad x_i \longleftrightarrow \partial B_{\varepsilon}(x_i), \forall i=1,\ldots,4 \big) \\
& \qquad + \mathbb{P}\big(x_1 \longleftrightarrow x_2 \centernot{\longleftrightarrow} x_3 \longleftrightarrow x_4 \, \vert \, x_i \longleftrightarrow \partial B_{\varepsilon}(x_i), \forall i=1,\ldots,4 \big) \\
& \qquad + \mathbb{P}\big(x_1 \longleftrightarrow x_3 \centernot{\longleftrightarrow} x_2 \longleftrightarrow x_4 \, \vert \, x_i \longleftrightarrow \partial B_{\varepsilon}(x_i), \forall i=1,\ldots,4 \big) \\
& \qquad + \mathbb{P}\big(x_1 \longleftrightarrow x_4 \centernot{\longleftrightarrow} x_2 \longleftrightarrow x_3 \, \vert \, x_i \longleftrightarrow \partial B_{\varepsilon}(x_i), \forall i=1,\ldots,4 \big)\big] \\
& = \varepsilon^{-5/12} \, \big[ \mathbb{P}\big(x_3 \longleftrightarrow x_4 \longleftrightarrow x_1 \, \vert \, x_1 \longleftrightarrow x_2; x_i \longleftrightarrow \partial B_{\varepsilon}(x_i), \forall i \big) \\
& \quad + \mathbb{P}\big(x_3 \longleftrightarrow x_4 \centernot{\longleftrightarrow} x_1 \, \vert \, x_1 \longleftrightarrow x_2; x_i \longleftrightarrow \partial B_{\varepsilon}(x_i), \forall i \big) \big] \\
& \qquad \qquad \mathbb{P}\big( x_1 \longleftrightarrow x_2 \, \vert \, x_i \longleftrightarrow \partial B_{\varepsilon}(x_i), \forall i \big) \\
& \quad + \varepsilon^{-5/12} \big[ \mathbb{P}\big(x_1 \longleftrightarrow x_3 \centernot{\longleftrightarrow} x_2 \longleftrightarrow x_4 \, \vert \, x_i \longleftrightarrow \partial B_{\varepsilon}(x_i), \forall i \big) \\
& \quad + \mathbb{P}\big(x_1 \longleftrightarrow x_4 \centernot{\longleftrightarrow} x_2 \longleftrightarrow x_3 \, \vert \, x_i \longleftrightarrow \partial B_{\varepsilon}(x_i), \forall i \big)\big] \\
& = \varepsilon^{-5/24} \, \mathbb{P}\big(x_3 \longleftrightarrow x_4 \, \vert \, x_1 \longleftrightarrow x_2; x_i \longleftrightarrow \partial B_{\varepsilon}(x_i), \forall i \big) \\
& \qquad \varepsilon^{-5/24} \mathbb{P}\big( x_1 \longleftrightarrow x_2 \, \vert \, x_i \longleftrightarrow \partial B_{\varepsilon}(x_i), \forall i \big) \\
& \quad + \varepsilon^{-5/12} \big[ \mathbb{P}\big(x_1 \longleftrightarrow x_3 \centernot{\longleftrightarrow} x_2 \longleftrightarrow x_4 \, \vert \, x_i \longleftrightarrow \partial B_{\varepsilon}(x_i), \forall i \big) \\
& \quad + \mathbb{P}\big(x_1 \longleftrightarrow x_4 \centernot{\longleftrightarrow} x_2 \longleftrightarrow x_3 \, \vert \, x_i \longleftrightarrow \partial B_{\varepsilon}(x_i), \forall i \big)\big].
\end{split}
\end{align}

Now imagine a situation in which $x_1$ and $x_2$ are very close to each other and far from $x_3$ and $x_4$. Then, the first term after the last equal sign is close~to
\begin{align}
\begin{split}
& P_2(x_1,x_2) P_2(x_3,x_4) = \mathrm{C}_{2}(x_1,x_2) \mathrm{C}_{2}(x_3,x_4) \\
& \qquad \qquad = C_2^2 \Vert x_1-x_2 \Vert^{-5/24} \Vert x_3-x_4 \Vert^{-5/24}.
\end{split}
\end{align}
The two remaining terms involve a four-arm event within the annulus $A_{\frac{\Vert x_1-x_2 \Vert}{2},1}(\frac{x_1+x_2}{2})$, that is, the event that four disjoint portions of loops (or possibly of the same loop) cross the annulus $A_{\frac{\Vert x_1-x_2 \Vert}{2},1}(\frac{x_1+x_2}{2})$ from the circle $\partial B_{\frac{\Vert x_1-x_2 \Vert}{2}}(\frac{x_1+x_2}{2})$ to $\partial B_1(\frac{x_1+x_2}{2})$, with alternating orientations (see Figure \ref{fig-pivotal}). This event has probability of order $\Vert x_1-x_2 \Vert^{5/4}$ (see, e.g., Remark 4.10 of \cite{GPS13}) and it needs to happen if $x_1$ and $x_2$ do not belong to the same continuum cluster.

\begin{figure}[!ht]
	\begin{center}
		\includegraphics[width=5cm]{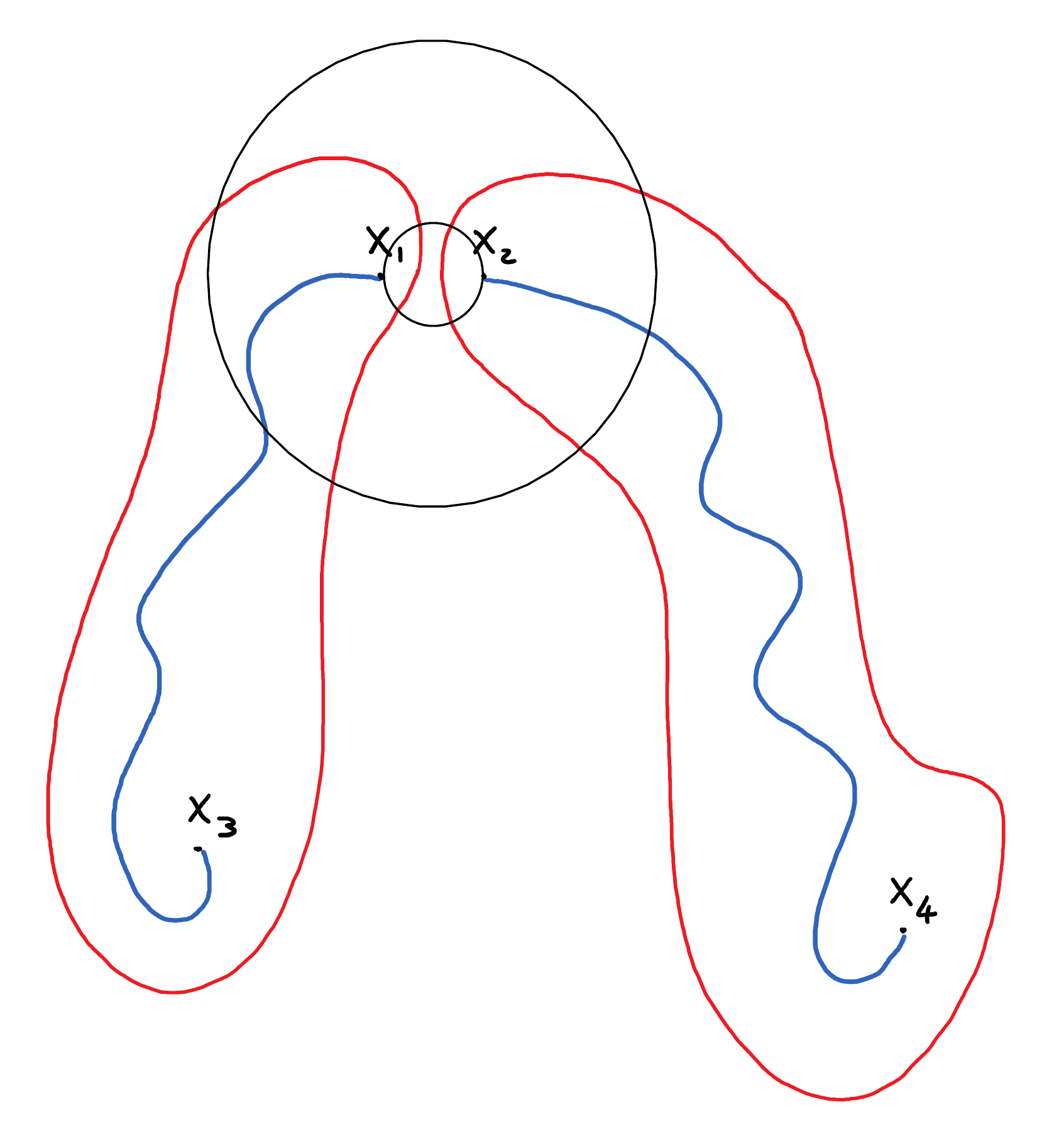}
		\caption{Sketch of one of the configurations contributing to the four-point function $\mathrm{C}_4(x_1,\ldots,x_4)$. The configuration depicted involves the four-arm event $\mathcal{F}_{\frac{\Vert x_1-x_2 \Vert}{2},1}(\frac{x_1+x_2}{2})$. The two loops wind around two (continuum) open clusters which support open paths (the wiggly lines) between $x_1$ and $x_3$ and between $x_2$ and $x_4$, respectively.}
		\label{fig-pivotal}
	\end{center}
\end{figure}

If we let $\mathcal{F}_{\frac{\Vert x_1-x_2 \Vert}{2},1}(\frac{x_1+x_2}{2})$ denote the four-arm event describe above, we have
\begin{align}
\begin{split}
& \varepsilon^{-5/12} \mathbb{P}\big(x_1 \longleftrightarrow x_3 \centernot{\longleftrightarrow} x_2 \longleftrightarrow x_4 \, \vert \, x_i \longleftrightarrow \partial B_{\varepsilon}(x_i), \forall i=1,\ldots,4 \big) \\
& = \varepsilon^{-5/24} \mathbb{P}\Big(x_1 \longleftrightarrow x_3 \centernot{\longleftrightarrow} x_2 \longleftrightarrow x_4 \, \vert \\
& \qquad \qquad \quad \mathcal{F}_{\frac{\Vert x_1-x_2 \Vert}{2},1}\Big(\frac{x_1+x_2}{2}\Big); x_i \longleftrightarrow \partial B_{\varepsilon}(x_i), \forall i=1,\ldots,4 \Big) \\
& \qquad \varepsilon^{-5/24} \mathbb{P}\Big( \mathcal{F}_{\frac{\Vert x_1-x_2 \Vert}{2},1}\Big(\frac{x_1+x_2}{2}\Big) \, \vert \, x_i \longleftrightarrow \partial B_{\varepsilon}(x_i), \forall i=1,2 \Big)
\end{split}
\end{align}
and a similar expression when $x_1 \longleftrightarrow x_4 \centernot\longleftrightarrow x_2 \longleftrightarrow x_3$.

We now take $\varepsilon$ of order $\Vert x_1-x_2 \Vert$. Considering that $\mathcal{F}_{\frac{\Vert x_1-x_2 \Vert}{2},1}(\frac{x_1+x_2}{2})$ has probability of order $\Vert x_1-x_2 \Vert^{5/4}$, we see that
\begin{align}
\begin{split}
& \varepsilon^{-5/24} \mathbb{P}\Big( \mathcal{F}_{\frac{\Vert x_1-x_2 \Vert}{2},1}\Big(\frac{x_1+x_2}{2}\Big) \, \vert \, x_i \longleftrightarrow \partial B_{\varepsilon}(x_i), \forall i=1,2 \Big) \\
& \qquad \sim \Vert x_1-x_2 \Vert^{-5/24} \Vert x_1-x_2 \Vert^{5/4}.
\end{split}
\end{align}

Putting all these observations together, we arrive at
\begin{align}
\begin{split} \label{eq:proto-OPE}
& \mathrm{C}_{4}(x_1,\ldots,x_4) \\
& \qquad \sim \Vert x_1-x_2 \Vert^{-5/24} \Big( \Vert x_3-x_4 \Vert^{-5/24} + F(x;x_3,x_4) \Vert x_1-x_2 \Vert^{5/4} \Big),
\end{split}
\end{align}
for $x=\frac{x_1+x_2}{2}$ and some function $F(x;x_3,x_4)$, which can be expressed in terms of conditional probabilities, where the notation reflects the fact that, for $x_1$ and $x_2$ very close to each other, these conditional probabilities are essentially a function of $x$ rather than $x_1$ and $x_2$ separately.

If we now introduce the suggestive notation (standard in the physics literature)
\begin{align}
\begin{split}
\mathrm{C}_{4}(x_1,\ldots,x_4) = \langle \Phi(x_1) \ldots \Phi(x_4) \rangle,
\end{split}
\end{align}
assuming that the field $\Phi$ is \emph{canonically normalized}, so that $\langle \Phi(x_1)\Phi(x_2) \rangle = \Vert x_1-x_2 \vert^{-5/24}$, we can write \eqref{eq:proto-OPE} as
\begin{align}
\begin{split} \label{eq:proto-OPE-bis}
& \langle \Phi(x_1) \ldots \Phi(x_4) \rangle \\
& \qquad \sim \Vert x_1-x_2 \Vert^{-5/24} \Big( \langle \Phi(x_3)\Phi(x_4) \rangle + F(x;x_3,x_4) \Vert x_1-x_2 \Vert^{5/4} \Big).
\end{split}
\end{align}
This equation suggests the presence of a new primary field of scaling dimension $5/4$, which can be identified with the so-called ``four-leg'' operator (see \cite{VJS12} and reference therein). An operator with scaling dimension $5/4$ in the percolation CFT was identified by Dotsenko \cite{Dotsenko16} using Coulomb gas techniques and taking the $q \to 1$ limit of the $q$-state Potts model, but without providing an interpretation in terms of a percolation event.

Equation \eqref{eq:proto-OPE-bis} is suggestive of the operator product expansion (OPE)
\begin{align} \label{eq:OPE}
\Phi(x_1)\Phi(x_2) = \Vert x_1-x_2 \Vert^{-5/24} \Big( \text{Id} + C_{\Phi,\Phi}^{\mathcal{F}} \mathcal{F}(x) \Vert x_1-x_2 \Vert^{5/4} + \ldots \Big),
\end{align}
where Id denotes the identity operator, $\mathcal{F}$ is a field of scaling dimension $5/4$ and $C_{\Phi,\Phi}^{\mathcal{F}}$ is the \emph{structure constant} appearing in the three-point function $\langle \Phi(x_1)\Phi(x_2)\mathcal{F}(x_3) \rangle$.

An OPE is meaningful only when an expectation is taken on both sides of the equation, and the equal sign means that one can multiply each side of the equation by any combination of conformal fields before taking the expectation (see, e.g., \cite{DFMS}). While the derivation of \eqref{eq:proto-OPE-bis} does not prove the validity of \eqref{eq:OPE}, it provides strong support for it and a direct point of contact with the physics literature. 

\medskip

\section{Scaling limit of the lattice field} \label{sec:field-s-lim}

In this section we discuss the scaling limit of the lattice field \eqref{def:lattice-field} and prove Theorem \ref{thm:Sobolev-conv-field} and Corollary \ref{cor:conf-cov}. Before we can present the proofs, we need to introduce some more terminology and present some auxiliary results.

Let $M^a=\{\mu^a_i\}_i$ denote the collection of normalized counting measures \eqref{def:counting-measure} of the open clusters of critical percolation on $a\mathcal{T}$ introduced in Section \ref{sec:main-results}. It follows from \cite{CCK19} (see Theorems 2 and 3) that $M^a$ converges in distribution, as $a \to 0$, to a random, countable collection, $M=\{ \mu_k \}_k$, of measures of bounded support, in the topology induced by the distance function
\begin{align}\label{eq:metric}
\begin{split}
& \text{Dist}_{\text{P}}(M_1, M_{2}) \\
& \quad :=\inf\{\epsilon>0: \forall \mu_1 \in M_1~\exists \mu_2 \in M_2 \text{ s.t. }\text{d}_{\text{P}}(\mu_1,\mu_2)\leq\epsilon \text{ and vice versa}\},
\end{split}
\end{align}
where $M_1$ and $M_2$ are collections of measures and $\text{d}_{\text{P}}$ is the Prokhorov distance between measures.

The limiting collection $M$ is invariant in distribution under translations and rotations and transforms covariantly under scale transformations (see Theorem 4 of \cite{CCK19}) in the sense that, formally, for any $s>0$, $\mu_k(sx)=s^{-5/48}\mu_k(x)$. More precisely, the collection $M_s=\{(\mu_s)_k\}_k$ of measures defined by
\begin{align}
(\mu_s)_k(f) := \int_{{\mathbb R}^2} f\Big(\frac{x}{s}\Big) d\mu_k(x)
\end{align}
has the same law as the collection $\{s^{2-5/48}\mu_k\}_k$. In particular, for any $L,s>0$, the distribution of $\{\mu_k(\mathbf{1}_{[-sL,sL]^2})\}_k$ is the same as the distribution of $\{s^{91/48}\mu_k(\mathbf{1}_{[-L,L]^2})\}_k$.

Moreover, combining Theorem 13, Lemma 9 and Theorem 3 of \cite{CCK19} implies that $M$ is measurable with respect to the full scaling limit $\Lambda$ of the collection of critical percolation interfaces on $a\mathcal{H}$ constructed in \cite{CN06}. For this reason, with a slight abuse of notation, we can use $\mathbb{P}$ to denote the distribution of $M$. 

Combining these results with \eqref{eq:lattice-field}, it is tempting to try to define a field
\begin{equation} \label{eq:cont-field}
`` \, \Phi(f) = \sum_{k} \sigma_k \mu_k(f), "
\end{equation}
where $\{\sigma_k\}_k$ is a collection of independent, symmetric, $(\pm 1)$-valued random variables assigned to the measures $\{\mu_k\}_k$. However, due to scale invariance, even for functions $f$ of bounded support, the sum above contains infinitely many terms, and the scaling properties of the $\mu_k$'s suggest that the collection $M$ may in general not be absolutely summable.

In order to make sense of \eqref{eq:cont-field}, let $M_{n,\varepsilon}$ denote the collection of measures $\mu_k \in M$ such that $\supp(\mu_k) \cap [-n,n]^2 \neq \emptyset$ and $\diam(\supp(\mu_k))>\varepsilon$, where $\supp(\mu_k)$ denotes the support of $\mu_k$ and $\text{diam}$ denotes the Euclidean diameter. 
One can show, for example by applying the proof of Proposition 2.2 of \cite{CN09} to the easier case of Bernoulli percolation, that for any $n \in \mathbb{N}$ and any $\varepsilon>0$, the cardinality of $M_{n,\varepsilon}$ is finite with probability one. Thanks to this observation, we can define the $(n,\varepsilon)$-cutoff field
\begin{equation} \label{eq:cutoff-cont-field}
\Phi_{n,\varepsilon} := \sum_{k: \mu_k \in M_{n,\varepsilon}} \sigma_k \, \mu_k \vert_{[-n,n]^2},
\end{equation}
where $\mu_k \vert_{[-n,n]^2}$ denotes the restriction of $\mu_k$ to $[-n,n]^2$.

Below we will show how one can remove the cutoffs $\varepsilon$ and $n$,
but before doing that, we need some more preliminaries. Let $\{ u_{i,j} \}_{i,j=1}^{\infty}$ with
\begin{align} \label{def:eigenfunctions}
& u_{i,j}(\xi,\zeta) := \frac{1}{n} \sin\left(\frac{\pi i (\xi+n)}{2n}\right) \sin\left(\frac{\pi j (\zeta+n)}{2n}\right),
\end{align}
for $-n \leq \xi \leq n, -n \leq \zeta \leq n$, denote the eigenfunctions of the negative Laplacian (i.e., $-\Delta$) on $[-n,n]^2$ with Dirichlet boundary condition, with eigenvalues 
\begin{align} \label{eq:eigenvalues}
\lambda_{i,j} = \frac{\pi^2}{4n^2}(i^2+j^2), \; i,j=1,2,\ldots
\end{align}
and $L^2$ norm $\Vert u_{i,j} \Vert_{L^2} = 1$.

The functions $\{ u_{i,j} \}_{i,j=1}^{\infty}$ form an orthonormal basis of $L^2([-n,n]^2)$ and of the Sobolev space $H_0^1([-n,n]^2)$, which is the closure of the space $C^{\infty}_0([-n,n]^2)$ of infinitely differentiable functions of compact support on $[-n,n]^2$ with respect to the norm
\begin{align}
\Vert f \Vert_{H^1_0}^2 := \int_{-\infty}^{\infty} \int_{-\infty}^{\infty} \vert \Delta f(\xi,\zeta) \vert^2 d\xi d\zeta,
\end{align}
and they satisfy $\Vert u_{i,j} \Vert_{H^1_0} = \lambda_{i,j}$. As a consequence, each $f \in H^1_0([-n,n]^2)$ has a unique orthogonal decomposition $f=\sum_{i,j} a_{i,j} u_{i,j}$ such that $\|f\|^2_{H^1_0} = \sum_{i,j} |a_{i,j}|^2 \lambda^2_{i,j}$.
Since $C^{\infty}_0([-n,n]^2) \subset H^1_0([-n,n]^2)$, the same holds for each $f \in C^{\infty}_0([-n,n]^2)$.
Moreover, if $f\in C_0^\infty([-n,n]^2)$, then
\begin{equation}\label{e:ww}
\sum_{i,j} |a_{i,j}|^2 \lambda_{i,j}^{2\alpha} <\infty, \quad \forall \alpha>0.
\end{equation}

To see why \eqref{e:ww} holds, one can assume without loss of generality that $\alpha\geq 1$ is an integer. For such an $\alpha$ and for every $f\in C_0^\infty([-n,n]^2)$, one has that $\Delta^{2\alpha} f\in C_0^\infty([-n,n]^2)$, and consequently $\Delta^{2\alpha} f = \sum_{i,j} \langle \Delta^{2\alpha} f , u_{i,j}\rangle_{L^2} u_{i,j}$, where the series converges in $L^2([-n,n]^2)$. Integration by parts yields moreover that $\langle \Delta^{2\alpha} f , u_{i,j}\rangle_{L^2} = \langle  f , \Delta^{2\alpha} u_{i,j}\rangle_{L^2}= \lambda_{i,j}^{2\alpha} \langle f , u_{i,j}\rangle_{L^2}$, from which we deduce that
\begin{equation}
\sum_{i,j} |a_{i,j}|^2 \lambda_{i,j}^{2\alpha} = \langle \Delta^{2\alpha} f, f\rangle_{L^2} \leq \|\Delta^{2\alpha} f\|_{L^2} \,\cdot  \|f\|_{L^2}<\infty,
\end{equation}
as claimed.

Given \eqref{e:ww}, $H_0^\alpha([-n,n]^2)$ is defined as the closure of $C_0^\infty([-n,n]^2)$ with respect to the norm $\|f\|^2_{H_0^{\alpha}} := \sum_{i,j} |a_{i,j}|^2 \lambda^{2\alpha}_{i,j}$. The Sobolev space $H^{-\alpha}([-n,n]^2)$ is then defined as the Hilbert dual of $H_0^\alpha([-n,n]^2)$, that is, the space of continuous linear functionals on $H_0^\alpha([-n,n]^2)$, endowed with the operator norm $\|h\|_{H^{-\alpha}} := \sup_{f \in H_0^\alpha([-n,n]^2): \|f\|_{H^{\alpha}_0} \leq 1} |h(f)|$. One has that $L^2([-n,n]^2) \subset H^{-\alpha}([-n,n]^2)$. Moreover, the action of $h=\sum_{i,j} a_{i,j} u_{i,j}\in L^2([-n,n]^2)$ on $f\in H_0^\alpha([-n,n]^2)$ is given by $h(f) = \int_{[-n,n]^2} h(x) f(x)\, dx$ and
\begin{align} \label{eq:norm}
\|h\|^2_{H^{-\alpha}} = \sum_{i,j} \lambda^{-2\alpha}_{i,j} |a_{i,j}|^2.
\end{align}

Lastly, let $D_n\equiv[-n,n]^2$ and consider the functions
\begin{align}
\hat{\Phi}^a_{n}(x) = a^2 \pi_a^{-1} \sum_{i:\mathcal{C}^a_i \in \mathscr{C}^a_{D_n}} \sigma_i \sum_{x^a \in \mathcal{C}^a_i \cap D_n} \frac{\mathbf{1}_{x^a}(x)}{A_a}, 
\end{align}
introduced earlier, and
\begin{align} \label{def:lattice-smooth-field-cutoff}
\hat{\Phi}^a_{n,\varepsilon}(x) := a^2 \pi_a^{-1} \sum_{i:\mathcal{C}^a_i \in \mathscr{C}^a_{D_n}, \diam(\mathcal{C}^a_i)>\varepsilon} \sigma_i \sum_{x^a \in \mathcal{C}^a_i \cap D_n} \frac{\mathbf{1}_{x^a}(x)}{A_a}, 
\end{align}
where $\mathbf{1}_{x^a}$ is the indicator function of the elementary hexagon $x_a$ of $a\mathcal{H}$ (the hexagon centered at $x^a \in a\mathcal{T}$---recall that, with a slight abuse of notation, we use $x^a$ both for elementary hexagons of $a\mathcal{H}$ and their centers in $a\mathcal{T}$) and $A_a$ denotes its area.

In the next theorem and in the rest of the paper, we will use $E^a$ to denote expectation with respect to $P^a$ and $\langle \cdot \rangle$ to denote expectation with respect to the distribution of $M = \{ \mu_k \}_k$ and of the random signs $\{\sigma_k\}_k$ assigned to the measures $\mu_k \in M$.

\begin{theorem} \label{theorem:Sobolev_convergence}
	For every $n \in \mathbb{N}$ and $\varepsilon>0$, as $a \to 0$, $\hat\Phi^a_n$ and $\hat\Phi^a_{n,\varepsilon}$ converge in distribution to two random elements of the Sobolev space $H^{-1}([-n,n]^2)$, $\hat\Phi_n$ and $\hat\Phi_{n,\varepsilon}$, respectively. The convergence is in the topology induced by $\Vert \cdot \Vert_{H^{-1}}$ and the limits are such that $\big\langle\Vert \hat\Phi_{n} \Vert^2_{H^{-1}}\big\rangle < \infty$ and $\big\langle\Vert \hat\Phi_{n,\varepsilon} \Vert^2_{H^{-1}}\big\rangle < \infty$. Moreover, $\hat\Phi_{n,\varepsilon}$ coincides with $\Phi_{n,\varepsilon}$ in distribution on $C^{\infty}_0([-n,n]^2)$, and 
	\begin{align}
	\lim_{\varepsilon \to 0} \big\langle\| \hat\Phi_{n} - \hat\Phi_{n,\varepsilon} \|^2_{H^{-1}}\big\rangle = 0.
	\end{align}
\end{theorem}
\begin{proof}
	Since $\hat{\Phi}^a_{n}, \hat{\Phi}^a_{n,\varepsilon} \in L^2(D_n)$, we can think of them as elements of $H^{-\alpha}(D_n)$ and apply \eqref{eq:norm}.
	Given $\alpha>1/2$, let $\epsilon=\alpha-1/2>0$ and $\alpha'=1/2+\epsilon/2<\alpha$; then \eqref{def:eigenfunctions}, Theorem \ref{thm:scal-lim-connection-probabilities} and \eqref{eq:2-point-function} imply that
	\begin{align}
	\begin{split} \label{eq:bounded-second-moment}
	& \limsup_{a \to 0} \big\langle \Vert \hat{\Phi}^a_{n,\varepsilon} \Vert^2_{H^{-\alpha'}} \big\rangle^a = \limsup_{a \to 0} \sum_{i,j} \frac{1}{\lambda_{i,j}^{2\alpha'}} \Big\langle \big( \hat{\Phi}^a_{n,\varepsilon}(u_{i,j}) \big)^2 \Big\rangle^a \\
	& \leq \left(\sum_{i,j} \frac{\| u_{i,j} \|^2_{L^{\infty}}}{\lambda_{i,j}^{2\alpha'}}\right) \limsup_{a \to 0} a^4 \pi_a^{-2} \left\langle\Big( \sum_{i:\mathcal{C}^a_i \in \mathscr{C}^a_{D_n}, \diam(\mathcal{C}^a_i)>\varepsilon} \sigma_i \vert \mathcal{C}^a_i \cap D_n \vert\Big)^2\right\rangle^a \\
	& \leq \frac{1}{n^2}\left(\sum_{i,j} \frac{1}{\lambda_{i,j}^{1+\epsilon}}\right) \limsup_{a \to 0} a^4 \pi_a^{-2} E^a\Big(\sum_{i:\mathcal{C}^a_i \in \mathscr{C}^a_{D_n}, \diam(\mathcal{C}^a_i)>\varepsilon} \vert \mathcal{C}^a_i \cap D_n \vert^2\Big) \\
	& \leq \frac{1}{n^2} \left(\sum_{i,j} \frac{1}{\lambda_{i,j}^{1+\epsilon}}\right) \limsup_{a \to 0} a^4 \pi_a^{-2} E^a\Big(\sum_{i:\mathcal{C}^a_i \in \mathscr{C}^a_{D_n}} \vert \mathcal{C}^a_i \cap D_n \vert^2\Big) \\
	& = \frac{1}{n^2} \left(\sum_{i,j} \frac{1}{\lambda_{i,j}^{1+\epsilon}}\right) \limsup_{a \to 0} a^4 \sum_{x^a_1,x^a_2 \in D^a_n} \pi_a^{-2} P^a_2(x^a_1,x^a_2) \\
	& = \frac{C_2}{n^2} \left(\sum_{i,j} \frac{1}{\lambda_{i,j}^{1+\epsilon}}\right) \int_{D_n}\int_{D_n} \vert x_1-x_2 \vert^{-5/24} dx_1 dx_2 \\
	& = C_2 \Big(\frac{2}{\pi}\Big)^{2(1+\epsilon)} n^{2\epsilon} \left(\sum_{i,j=1}^{\infty} \frac{1}{(i^2+j^2)^{1+\epsilon}}\right) \int_{D_n}\int_{D_n} \frac{dx_1 dx_2}{\vert x_1-x_2 \vert^{5/24}} < \infty,
	\end{split}
	\end{align}
	where, in the last equality, we have used \eqref{eq:eigenvalues}.
	
	On the fourth line of the above calculation, we dropped the condition $\diam(\mathcal{C}^a_i)>\varepsilon$, which distinguishes $\Phi^a_{n,\varepsilon}$ from $\Phi^a_{n}$, so the final upper bound applies also to $\Phi^a_{n}$:
	\begin{align}
	\begin{split} \label{eq:finite-limsup}
	& \limsup_{a \to 0} \big\langle \Vert \hat{\Phi}^a_{n} \Vert^2_{H^{-\alpha'}} \big\rangle^a = \limsup_{a \to 0} \sum_{i,j} \frac{1}{\lambda_{i,j}^{2\alpha'}} \Big\langle \big( \hat{\Phi}^a_{n}(u_{i,j}) \big)^2 \Big\rangle^a \\
	& \quad \leq C_2 \Big(\frac{2}{\pi}\Big)^{2(1+\epsilon)} n^{2\epsilon} \left(\sum_{i,j=1}^{\infty} \frac{1}{(i^2+j^2)^{1+\epsilon}}\right) \int_{D_n}\int_{D_n} \frac{dx_1 dx_2}{\vert x_1-x_2 \vert^{5/24}} < \infty.
	\end{split}
	\end{align}
	
	This, combined with Chebyshev's inequality, implies that $\hat\Phi^a_{n}$ and $\hat\Phi^a_{n,\varepsilon}$ are tight, as $a \to 0$, in $H^{-\alpha}(D_n)$ for $\alpha>1/2$. Moreover, Rellich's theorem implies that $H^{-\alpha_1}(D_n)$ is compactly embedded in $H^{-\alpha_2}(D_n)$ for any $\alpha_1<\alpha_2$ and thus, in particular, that the closure of a ball of finite radius in $H^{-\alpha'}(D_n)$ is compact in $H^{-\alpha}(D_n)$. Therefore, $\hat\Phi^a_{n}$ and $\hat\Phi^a_{n,\varepsilon}$ have subsequential limits in distribution in $H^{-\alpha}(D_n)$ for any $\alpha>1/2$. Furthermore, since $\hat\Phi^a_{n}$ and $\hat\Phi^a_{n,\varepsilon}$ are naturally coupled via the percolation model, one has joint convergence in distribution of $(\hat\Phi^{a_k}_{n},\hat\Phi^{a_k}_{n,\varepsilon})$ along some sequence $a_k \to 0$. We denote the limit by $(\hat\Phi_{n},\hat\Phi_{n,\varepsilon})$.
	
	We now fix $\alpha=1$ and note that the Sobolev space $H^{-1}(D_n)$ endowed with the norm $\Vert \cdot \Vert_{H^{-1}}$ is a complete separable metric space. Therefore, using Skorokhod's representation theorem, we can find coupled versions of $(\hat\Phi^{a_k}_{n},\hat\Phi^{a_k}_{n,\varepsilon})$ and $(\hat\Phi_{n},\hat\Phi_{n,\varepsilon})$ such that $\lim_{k \to \infty}\Vert \hat\Phi_{n}-\hat\Phi^{a_k}_{n} \Vert_{H^{-1}}=0$ and $\lim_{k \to \infty}\Vert \hat\Phi_{n,\varepsilon}-\hat\Phi^{a_k}_{n,\varepsilon} \Vert_{H^{-1}}=0$ almost surely.
	
	This implies that $\Vert \Phi^{a_k}_n \Vert^2_{H^{-1}}$ converges to $\Vert \hat\Phi_n \Vert^2_{H^{-1}}$ almost surely. In addition, Fatou's lemma and \eqref{eq:finite-limsup} imply that
	\begin{align}
	\big\langle \Vert \hat\Phi_n \Vert^2_{H^{-1}} \big\rangle \leq \limsup_{k \to \infty} \big\langle \Vert \hat{\Phi}^{a_k}_{n} \Vert^2_{H^{-1}} \big\rangle^{a_k} \leq \hat{C}_n
	\end{align}
	for some $\hat{C}_n<\infty$. Similar considerations, using equation \eqref{eq:bounded-second-moment}, give $\big\langle \Vert \hat\Phi_{n,\varepsilon} \Vert^2_{H^{-1}} \big\rangle \leq \hat{C}_n$.
	
	Now let $Y^{\varepsilon}_k \equiv \Vert \hat\Phi^{a_k}_{n} - \hat\Phi^{a_k}_{n,\varepsilon} \Vert^2_{H^{-1}}$ and $Y^{\varepsilon} \equiv \Vert \hat\Phi_{n} - \hat\Phi_{n,\varepsilon} \Vert^2_{H^{-1}}$ and observe that $Y^{\varepsilon}_k$ converges to $Y^{\varepsilon}$ almost surely due to the almost sure convergence of $(\hat\Phi^{a_k}_{n},\hat\Phi^{a_k}_{n,\varepsilon})$ to $(\hat\Phi_{n},\hat\Phi_{n,\varepsilon})$ with the norm $\Vert \cdot \Vert_{H^{-1}}$. Therefore, applying Fatou's lemma again, we have that $\langle Y^{\varepsilon} \rangle \leq \limsup_{k \to \infty} \langle Y^{\varepsilon}_k \rangle$, which means that
	\begin{align}
	\begin{split} \label{eq:L2-bound}
	& \big\langle\| \hat\Phi_{n} - \hat\Phi_{n,\varepsilon} \|^2_{H^{-1}}\big\rangle \leq \limsup_{k \to \infty} \big\langle\| \hat\Phi^{a_k}_{n} - \hat\Phi^{a_k}_{n,\varepsilon} \|^2_{H^{-1}}\big\rangle^{a_k}.
	\end{split}
	\end{align}
	
	Using again the fact that $\hat\Phi^a_{n}, \hat\Phi^a_{n,\varepsilon} \in L^2(D_n)$, a calculation similar to \eqref{eq:bounded-second-moment} shows that, for any $\varepsilon>\varepsilon'>0$ and some constants $C,C'<\infty$,
	\begin{align}
	\begin{split} \label{eq:limsup-dist}
	& \limsup_{k \to \infty} \big\langle \Vert \hat{\Phi}^{a_k}_{n}  - \hat{\Phi}^{a_k}_{n,\varepsilon}\Vert^2_{H^{-1}} \big\rangle^{a_k} = \limsup_{k \to \infty} \sum_{i,j} \frac{1}{\lambda_{i,j}^{2}} \Big\langle \Big[ \big( \hat{\Phi}^{a_k}_{n} - \hat{\Phi}^{a_k}_{n,\varepsilon} \big) (u_{i,j}) \big]^2 \Big\rangle^{a_k} \\
	& \leq \left(\sum_{i,j} \frac{\| u_{i,j} \|^2_{L^{\infty}}}{\lambda_{i,j}^{2}}\right) \limsup_{k \to \infty} a_k^4 \pi_{a_k}^{-2} \left\langle\Big( \sum_{i:\mathcal{C}^{a_k}_i \in \mathscr{C}^{a_k}_{D_n}, \diam(\mathcal{C}^{a_k}_i)\leq\varepsilon} \sigma_i \vert \mathcal{C}^{a_k}_i \cap D_n \vert\Big)^2\right\rangle^{a_k} \\
	& \leq \frac{1}{n^2}\left(\sum_{i,j} \frac{1}{\lambda_{i,j}^{2\alpha}}\right) \limsup_{k \to \infty} a_k^4 \pi_{a_k}^{-2} E^{a_k}\Big(\sum_{i:\mathcal{C}^{a_k}_i \in \mathscr{C}^{a_k}_{D_n}, \diam(\mathcal{C}^{a_k}_i)\leq\varepsilon} \vert \mathcal{C}^{a_k}_i \cap D_n \vert^2\Big) \\
	& \leq \frac{1}{n^2}\left(\sum_{i,j} \frac{1}{\lambda_{i,j}^{2}}\right) \limsup_{k \to \infty} a_k^4 \sum_{x_1^{a_k},x_2^{a_k}\in{a_k\mathcal{T} \cap D_n: |x^{a_k}_1-x^{a_k}_2|\leq\epsilon}} \pi_{a_k}^{-2} P^{a_k}_2(x^{a_k}_1,x^{a_k}_2) \\
	& \leq \frac{C}{n^2}\left(\sum_{i,j} \frac{1}{\lambda_{i,j}^{2}}\right) \int_0^{\varepsilon} r^{1-5/24} dr \\
	& \leq C' n^2 \left(\sum_{i,j=1}^{\infty} \frac{1}{(i^2+j^2)^{2}}\right) \varepsilon^{43/24},
	\end{split}
	\end{align}
	where, in the last inequality, we have used \eqref{eq:eigenvalues}.
	
	Combined with \eqref{eq:L2-bound}, \eqref{eq:limsup-dist} implies that
	\begin{align}
	\begin{split} \label{eq:mean-square-limit}
	\lim_{\varepsilon \to 0} \big\langle\| \hat\Phi_{n} - \hat\Phi_{n,\varepsilon} \|^2_{H^{-1}}\big\rangle = 0,
	\end{split}
	\end{align}
	showing that $\hat\Phi_{n,\varepsilon}$ converges in mean square to $\hat\Phi_{n}$, as $\varepsilon \to 0$, with the norm $\Vert \cdot \Vert_{H^{-1}}$.
	
	We will show next that $\hat\Phi^a_{n,\varepsilon}$ has a unique limit in $H^{-1}(D_n)$, in the topology induced by $\Vert \cdot \Vert_{H^{-1}}$, as $a \to 0$. 
	If $f \in C^{\infty}_0(D_n)$,
	\begin{align}
	\begin{split} \label{eq:tildephi=phi}
	& \hat\Phi^a_{n,\varepsilon}(f) = \int_{D_n} \hat\Phi^a_{n,\varepsilon}(x) f(x) dx \\
	& = \qquad a^2 \pi_a^{-1} \sum_{i:\mathcal{C}^a_i \in \mathscr{C}^a_{D_n}, \diam(\mathcal{C}^a_i)>\varepsilon} \sigma_i \sum_{x^a \in \mathcal{C}^a_i} \int_{D_n} \frac{\mathbf{1}_{x^a}(x)}{A_a} \, \big[f(x_a) + O(a)\big] dx \\
	& = \qquad \sum_{i:\mathcal{C}^a_i \in \mathscr{C}^a_{D_n},  \diam(\mathcal{C}^a_i)>\varepsilon} \sigma_i \Big[ a^2 \pi_a^{-1} \sum_{x^a \in \mathcal{C}^a_i \cap D_n} f(x_a) \Big] + a^2 \pi_a^{-1} R_f(a) \\
	& = \qquad \sum_{i:\mathcal{C}^a_i \in \mathscr{C}^a_{D_n}, \diam(\mathcal{C}^a_i)>\varepsilon} \sigma_i \, \mu^a_i(f) + a^2 \pi_a^{-1} R_f(a),
	\end{split}
	\end{align}
	where $R_f(a)$ is the sum of at most $K_n/a^2$ bounded terms, where $K_n$ ($\sim n^2$ as $n \to \infty$) is a constant, depending on $n$ but not on $a$, such that $K_n/a^2$ gives an upper bound for the number of vertices of $a\mathcal{T}$ in $a\mathcal{T} \cap D_n$. Each of the terms in $R_f(a)$ is of order $O(a)$, so that $\lim_{a \to 0} a^2 \pi_a^{-1} R_f(a) = 0$ because $\pi_a=a^{5/48+o(1)}$ as $a \to 0$ \cite{LSW02}.
	
	Therefore, by an application of Theorem~3 of \cite{CCK19}, as $a \to 0$, $\hat\Phi^a_{n,\varepsilon}(f)$ converges in distribution to $\Phi_{n,\varepsilon}(f)$. Since this is true for every $f \in C^{\infty}_0(D_n)$, all subsequential limits of $\hat\Phi^a_{n,\varepsilon}$ in $H^{-1}(D_n)$ in the topology induced by $\Vert \cdot \Vert_{H^{-1}}$ must coincide with $\Phi_{n,\varepsilon}$, in distribution, on $C^{\infty}_0(D_n)$.
	
	According to Lemma~A.5 of \cite{CGPR21}, the restriction of an element $F$ of $H^{-1}(D_n)$ to $C^{\infty}_0(D_n)$ determines the distribution of $F$ uniquely, so $\hat\Phi^a_{n,\varepsilon}$ has a unique limit $\hat\Phi_{n,\varepsilon}$ in $H^{-1}(D_n)$ in the topology induced by $\Vert \cdot \Vert_{H^{-1}}$. Moreover, $\hat\Phi_{n,\varepsilon}$ coincides with $\Phi_{n,\varepsilon}$, in distribution, on $C^{\infty}_0(D_n)$.
	
	The fact that $\hat\Phi_{n,\varepsilon}$ is unique, combined with the convergence of $\hat\Phi_{n,\varepsilon}$ to $\hat\Phi_{n}$ in mean square, \eqref{eq:mean-square-limit}, implies that
	$\hat\Phi^a_n$ has a unique limit in distribution in the topology induced by $\Vert \cdot \Vert_{H^{-1}}$ and concludes the proof.
\end{proof}

Let $\hat{\mathbb{P}}_n$ denote the distribution of the field $\hat\Phi_n$ from the previous theorem seen as an element of $H^{-1}([-n,n]^2)$. $\hat{\mathbb{P}}_n$ is a probability measure on the measurable space $(H^{-1}([-n,n]^2),\mathcal{B}_n)$, where $\mathcal{B}_n$ denotes the Borel sigma-algebra induced by the norm $\Vert \cdot \Vert_{H^{-1}}$.

\begin{lemma} \label{lemma:inf-vol-extension}
	The probability distributions $\{\hat{\mathbb{P}}_n\}_n$ have a unique extension to a distribution $\hat{\mathbb{P}}$ on $(H^{-1}(\mathbb{R}^2),\mathcal{B})$ where $\mathcal{B}$ denotes the Borel sigma-algebra induced by the norm $\Vert \cdot \Vert_{H^{-1}}$. In other words, there exists a field $\Phi$ defined on  $H_0^1(\mathbb{R}^2)$ such that $\Phi_n$ has the same distribution as $\Phi$ restricted to $[-n,n]^2$ (i.e., restricted to functions $f \in H_0^1([-n,n]^2)$).
\end{lemma}

\begin{proof}
	Let $D_n=[-n,n]^2$. It is clear from \eqref{eq:cutoff-cont-field} that, for any $\varepsilon>0$, for each $f \in C^{\infty}_0(D_n)$, $\Phi_{m,\varepsilon}(f)$ and $\Phi_{n,\varepsilon}(f)$ have the same distribution for all $m>n$. The same is true for $\hat\Phi_{m,\varepsilon}(f)$ and $\hat\Phi_{n,\varepsilon}(f)$ because $\hat\Phi_{k,\varepsilon}$ coincides in distribution with $\Phi_{k,\varepsilon}$ on $C^{\infty}_0(D_k)$ for every $k$ (see Theorem \ref{theorem:Sobolev_convergence}).
	
	Moreover, according to Theorem \ref{theorem:Sobolev_convergence}, as $\varepsilon \to 0$, $\hat\Phi_{k,\varepsilon}$ converges in mean square to $\hat\Phi_{k}$ in the topology induced by $\Vert \cdot \Vert_{H^{-1}}$. This, combined with the observation that $\vert \hat\Phi_{k,\varepsilon}(f) - \hat\Phi_{k}(f) \vert \leq \Vert f \Vert_{H^1_0} \Vert \big(\hat\Phi_{k,\varepsilon}-\hat\Phi_{k}\big) \Vert_{H^{-1}}$, implies that $\hat\Phi_{k,\varepsilon}(f)$ converges in mean square to $\hat\Phi_{k}(f)$,  as $\varepsilon \to 0$.
	Therefore, we can conclude that $\hat\Phi_{m}(f)$ and $\hat\Phi_{n}(f)$ have the same distribution for all $m>n$ and all $f \in C^{\infty}_0(D_n)$.
	
	According to Lemma~A.5 of \cite{CGPR21}, the restriction of an element $F$ of $H^{-1}(D_n)$ to $C^{\infty}_0(D_n)$ determines the distribution of $F$ uniquely. As a consequence, for every $m>n$, $\hat\Phi_{n}$ and the restriction $\hat\Phi_{m} \vert_{D_n}$ of $\hat\Phi_{m}$ to $D_n$ (more precisely, to functions in $H^1_0(D_n)$) have the same distribution.
	
	Since the spaces $H^{-1}(D_n), H^{-1}(\mathbb{R}^2)$ endowed with the norm $\Vert \cdot \Vert_{H^{-1}}$ are complete separable metric spaces, 
	$(H^{-1}(D_n),\mathcal{B}_n)$ and $(H^{-1}(\mathbb{R}^2),\mathcal{B})$ are standard Borel spaces.
	Therefore, we can apply Kolmogorov's extension theorem and conclude that there is a unique probability measure $\hat{\mathbb{P}}$ on $(H^{-1}(\mathbb{R}^2),\mathcal{B})$ such that $\hat{\mathbb{P}} \vert_{D_n} = \hat{\mathbb{P}}_n$ for all $n \in \mathbb{N}$.
\end{proof}

We are now ready to prove Theorem \ref{thm:Sobolev-conv-field} and Corollary \ref{cor:conf-cov}.

\begin{proof}[Proof of Theorem \ref{thm:Sobolev-conv-field}]
	Let $\Phi$ be a random element of $H^{-1}(\mathbb{R}^2)$ distributed according to $\hat{\mathbb{P}}$ from Lemma \ref{lemma:inf-vol-extension}. Then, by Theorem \ref{theorem:Sobolev_convergence}, for any $n \in \mathbb{N}$, as $a \to 0$, $\hat\Phi_n^{a}$ converges in distribution to $\Phi \vert_{[-n,n]^2}$, the restriction of $\Phi$ to $[-n,n]^2$, in the topology induced by the norm $\Vert \cdot \Vert_{H^{-1}}$.
	
	Given a function $f \in C^{\infty}_0(\mathbb{R}^2)$, there exists $n \in \mathbb{N}$ such that $\supp(f) \subset [-n,n]^2$. Therefore, using Lemma \ref{lemma:inf-vol-extension} and Theorem \ref{theorem:Sobolev_convergence}, we have that $\hat\Phi_n(f)$ is equal in distribution to $\Phi \vert_{[-n,n]^2}(f)=\Phi(f)$ and that $\hat{\Phi}_{n,\varepsilon}(f)$ is equal in distribution to
	\begin{align}
	\Phi_{n,\varepsilon}(f) = \sum_{k:\diam(\supp(\mu_k))>\varepsilon} \sigma_k \, \mu_k(f).
	\end{align}
	Moreover, \eqref{eq:L2-bound} and \eqref{eq:limsup-dist} imply that
	\begin{align}
	\begin{split}
	& \Big\langle \Big\vert \hat\Phi_n(f) - \hat\Phi_{n,\varepsilon}(f) \Big\vert^2 \Big\rangle \leq \Vert f \Vert_{H^1}^2 \langle \Vert \hat\Phi_n - \hat\Phi_{n,\varepsilon} \Vert_{H^{-1}}^2 \rangle \\
	& \qquad \leq \Vert f \Vert_{H^1}^2 C' n^2 \left(\sum_{i,j=1}^{\infty} \frac{1}{(i^2+j^2)^{2}}\right) \varepsilon^{43/24}.
	\end{split}
	\end{align}
	Combining these observations, we obtain
	\begin{align}
	\begin{split}
	& \Big\langle \Big\vert \Phi(f) - \sum_{k:\diam(\supp(\mu_k))>\varepsilon} \sigma_k \, \mu_k(f) \Big\vert^2 \Big\rangle \\
	& \qquad \leq \Vert f \Vert_{H^1}^2 C' n^2 \left(\sum_{i,j=1}^{\infty} \frac{1}{(i^2+j^2)^{2}}\right) \varepsilon^{43/24},
	\end{split}
	\end{align}
	as desired. \end{proof}

\begin{proof}[Proof of Corollary \ref{cor:conf-cov}]
	Assume first that $f \in C^{\infty}_0(\mathbb{R}^2)$ and let
	\begin{align}
	\tilde\Phi_{\varepsilon}(f):=\sum_{k:\diam(\supp(\mu_k))>\varepsilon} \sigma_k \, \mu_k(f).
	\end{align}
	Then, if $h_s(x)=sx$ denotes a scale transformation, Theorems 3 and 4 of \cite{CCK19} imply that
	\begin{align}
	\begin{split}
	& \tilde\Phi_{\varepsilon,s}(f):=\int_{{\mathbb R}^2} f\Big(\frac{x}{s}\Big) \tilde\Phi_{\varepsilon}(x) dx = \tilde\Phi_{\varepsilon}(f \circ h_{1/s}) \\
	& \qquad = \sum_{k:\diam(\supp(\mu_k))>\varepsilon} \sigma_k \, \mu_k(f \circ h_{1/s})
	\end{split}
	\end{align}
	has the same distribution as $s^{2-5/48}\tilde\Phi_{\varepsilon}(f)$. By Theorem \ref{thm:Sobolev-conv-field}, $\Phi(f)$ is the $L^2$ limit of $\tilde\Phi_{\varepsilon}(f)$, as $\varepsilon \to 0$, therefore $\Phi(f \circ h_{1/s})=\Phi_s(f)$ is distributed like $s^{2-5/48}\Phi(f)$. Formally, we can write
	\begin{align}
	\Phi_s(f) = \int_{{\mathbb R}^2} f\Big(\frac{x}{s}\Big) \Phi(x) dx = s^2 \int_{{\mathbb R}^2} f(x) \Phi(sx) dx,
	\end{align}
	which implies that $\Phi(sx)$ is equal in distribution to $s^{-5/48}\Phi(x)$.
	
	If $f \in H^1_0(\mathbb{R}^2)$ is not in $C^{\infty}_0(\mathbb{R}^2)$, take a sequence of functions $f_n \in C^{\infty}_0(\mathbb{R}^2)$ converging to $f$ in the topology induced by $\Vert \cdot \Vert_{H^1_0}$. This can always be done because $H^1_0(\mathbb{R}^2)$ is the closure of $C^{\infty}_0(\mathbb{R}^2)$ with respect to $\Vert \cdot \Vert_{H^1_0}$. The continuity of $\Phi$ implies the desired result. \end{proof}

\medskip

\noindent{\bf Acknowledgments.}
The author thanks Rob van den Berg, Omar El Dakkak, Jianping Jiang and Chuck Newman for useful discussions, Gesualdo Delfino for an interesting correspondence, and an anonymous referee for a careful reading of the manuscript and for useful comments and suggestions.
The author is especially grateful to Rob van den Berg for a conversation that revealed a gap in a previous version of the paper.

\end{document}